\documentclass[12pt]{amsart}
\usepackage{amssymb,latexsym,amsfonts, amsthm, amsmath, amsrefs, mathtools,color}
\usepackage{array}
\usepackage{enumerate}
\usepackage{enumitem}
\usepackage{booktabs}
\usepackage{fancyhdr}
\usepackage{verbatim}
\usepackage[margin=1in,twoside=false]{geometry}
\usepackage{multicol}
\usepackage{longtable}
\usepackage{hyperref}
\usepackage{diagbox}

\usepackage[normalem]{ulem}
\newtheorem{bigthm}{Theorem}   
\newtheorem{thm}{Theorem}[section]   
\newtheorem{cor}[thm]{Corollary}     
\newtheorem{lem}[thm]{Lemma}         
\newtheorem{prop}[thm]{Proposition}  

\theoremstyle{definition} 
\newtheorem{defn}[thm]{Definition}   
   
\newtheorem{ex}[thm]{Example}        

\newcommand{\refT}[1]{Theorem~\ref{T:#1}}
\newcommand{\refC}[1]{Corollary~\ref{C:#1}}

\newcommand{\refP}[1]{Proposition~\ref{P:#1}}
\newcommand{\refPp}[1]{Theorem~\ref{P:#1}} 
\newcommand{\refD}[1]{Definition~\ref{D:#1}}
\newcommand{\refL}[1]{Lemma~\ref{L:#1}}

\newcommand{\refEx}[1]{Example~\ref{Ex:#1}}
\newcommand{\refF}[1]{Figure~\ref{F:#1}}

\newcommand{\refS}[1]{Section~\ref{S:#1}}

\def\bd{\Gamma}
\newcommand{\od}{\stackrel{\mbox {\tiny {def}}}{=}}
\def\br{\overline}
\newcommand{\sm}{\setminus}

\def\e{\ensuremath{\mathbf{e}}}
\newcolumntype{R}[1]{>{\raggedleft\let\newline\\\arraybackslash\hspace{0pt}}p{#1}}

\def\nd{stable}
\def\ndy{stability}
\def\Nd{Stable}

\def\subnd{generic}
\def\subndy{genericity}

\newcommand{\code}{\ensuremath{\operatorname{code}}}
\newcommand{\nerve}{\ensuremath{\operatorname{nerve}}}
\newcommand{\link}{\ensuremath{\operatorname{link}}}
\newcommand{\del}{\ensuremath{\operatorname{del}}}

\newcommand{\cham}{\ensuremath{\operatorname{cham}}}

\newcommand{\cC}{\ensuremath{\mathcal{C}}}
\newcommand{\C}{\ensuremath{\mathcal{C}}}
\newcommand{\cU}{\ensuremath{\mathcal{U}}}
\newcommand{\cH}{\ensuremath{\mathcal{H}}}

\newcommand{\R}{\mathbb R}

\newcommand{\ZZ}{\mathbb{Z}}
\def\cB{\mathcal{B}}

\def\a{\alpha}
\def\b{\beta}
\def\D{\Delta}

\def\s{\sigma}
\def\S{\Sigma}
\def\t{\tau}

\title{Hyperplane Neural Codes and the Polar Complex}
\author{Vladimir Itskov}
\author{Alex Kunin}
\author{Zvi Rosen}

\begin{document}

\begin{abstract}
Hyperplane codes are a class of convex codes that arise as the output of a one layer feed-forward neural network. Here we establish several natural properties of  \nd\ hyperplane codes in terms of the {\it polar complex} of the code, a simplicial complex associated to any combinatorial code. We prove that the polar complex of a \nd\ hyperplane code is shellable and show that most currently known properties of hyperplane codes follow from the shellability of the appropriate polar complex. 
\end{abstract}

\maketitle

\vspace{-8mm}
\setcounter{tocdepth}{1}
\renewcommand{\baselinestretch}{0.6}\normalsize
\tableofcontents
\renewcommand{\baselinestretch}{1.0}\normalsize

\vspace{-8mm}


\section{Introduction}

Combinatorial codes, i.e.\ subsets of the Boolean lattice, 
naturally arise as outputs of neural networks.  
 A {\it codeword}  $\sigma \subseteq  [n] 
\od \{1,\ldots,n\}$ represents an allowed subset 
of co-active neurons, while a {\it code} is a 
collection $\C \subseteq 2^{[n]}$ of codewords.  
Combinatorial codes in a number of areas of the brain are often {\it convex}, 
i.e.\ they arise as an intersection pattern of convex sets in a Euclidean space \cite{hippocampus,visual,entorhinal}.
The combinatorial code of a one-layer 
feedforward neural network is also {convex}, as it arises as 
the intersection patterns of half-spaces \cite{R62, GI14}. It is well-known that a two-layer feedforward network can approximate any measurable function \cite{Cybenko89,Hornik91}, and thus may produce any combinatorial code. In contrast, 
the codes of one-layer feedforward networks are not well-understood. The intersection 
lattices of affine hyperplane arrangements have been studied in the oriented matroid 
literature \cite{OM, AOM, COM}. However,  combinatorial codes contain less detailed information than oriented matroids, and the precise relationship is not clear. 
We are motivated by the following question:
How can one determine  if a given combinatorial code 
is realizable as the output of a one-layer feedforward neural network?

We study \nd\ hyperplane codes, codes that 
arise from the intersection patterns of  half-spaces that are 
stable under certain small perturbations.  The paper is organized as follows. Relevant background and definitions are provided in \refS{background}. In \refS{obstructions}, we establish a number of obstructions 
that prevent  a combinatorial code from being a \nd\ hyperplane code. 
In \refS{mainresults}, we show that all but one of the currently known obstructions 
to being a \nd\ hyperplane code are subsumed 
by the condition that the \emph{polar complex} of the code, defined in \refS{polarcomplex}, is shellable. Lastly, in \refS{algebra} we show how  techniques from commutative algebra can be used to computationally detect the presence of these obstructions.


\section{Background}\label{S:background}

\subsection{\Nd\ Hyperplane Codes}

We call a collection $\cU = \{U_i\}$ of $n$ subsets $U_i \subseteq X$ of a set $X$ an \emph{arrangement} $(\cU,X)$. Note that we do {\it not} require  that $\bigcup_{i\in[n]} U_i = X$.
\begin{defn}
	For $\s\subseteq[n]$, let $A^{\cU}_\sigma$ denote the \textit{atom} of $(\mathcal{U},X)$ 
	\[
		A^{\cU}_\sigma \od \Bigl( \bigcap_{i \in \sigma} U_i \Bigr) \setminus  \bigcup_{j \not \in \sigma} U_j \subseteq X, \qquad \text{ where } 
		A^{\cU}_\varnothing \od X\setminus \bigcup_{i\in [n] } U_i.
	\]
	The \emph{code} of the arrangement $(\cU,X)$ is defined as
	\[ \code(\mathcal U,X)  \od \{ \sigma \subseteq [n] \text{ such that }A^{\mathcal U}_\sigma \neq \varnothing \} \subseteq 2^{[n]}.
	\]
	A \emph{realization} of a code $\cC$ is an arrangement $(\cU,X)$ such that $\cC = \code(\cU,X)$.
	The {\em simplicial complex of the code}, denoted $\Delta(\cC)$, is the closure of
	$\cC$ under inclusion: 
	\[ \Delta(\cC) \od   \{ \tau \mid \tau \subseteq \sigma \text{ for some } \sigma  \in \cC \}.\]
\end{defn}

\medskip 
\noindent 
Note that for $\cC = \code(\cU, X)$, the simplicial
complex of the code is equal to the \emph{nerve} of the corresponding cover:
\[ \Delta(\code(\cU,X)) = \nerve(\cU) \od \Bigl\{ \sigma \subseteq [n]   \mid  \bigcap_{i\in \sigma} U_i\neq \varnothing \Bigr\}. \]
   
A natural class of codes that arises in the context of neural networks is the class of  hyperplane codes \cite{GI14}. A hyperplane is a level set $H= \{x\in\R^d \mid w \cdot x  - h=0\}$ of a non-constant affine function. An {\it oriented}  hyperplane partitions $\R^d$ into three pieces: $\R^d=H ^+ \sqcup H \sqcup H ^-$,  where $H^{\pm}$ are the open half-spaces,   e.g. $H ^+\od \{ x\in \mathbb R^d \mid w\cdot x -h  > 0 \}$.

\begin{defn}
	A code $\C\subseteq 2^{[n]}$ is a \textit{hyperplane code}, if there exists an open convex subset $X \subseteq \R^d$ and a collection $\cH = \{H_1^+,\dots,H_n^+\}$ of  open half-spaces such that $\cC = \code(\{H_i^+ \cap X\},X)$. With a slight  abuse of notation,  we denote this arrangement of subsets of $X$ by $(\cH,X)$, thus  $\code(\cH,X) = \code(\{H_i^+ \cap X\},X)$.
\end{defn}

Hyperplane codes are produced by one-layer feedforward neural 
networks \cite{GI14}, where the convex set $X$ 
is often  the positive orthant $\R^d_{\geq 0}$.
A well-behaved subset of hyperplane codes 
are the {\it \nd} hyperplane codes. Informally, these are 
  codes that are preserved under small 
perturbations of the hyperplanes and the convex set $X$. These perturbations correspond to perturbations of the parameters of the neural network \cite{R62}, i.e.\ the vectors $(w_i,h_i) \in\R^d\times\R$ in our context. Thus, we restrict our attention to the class of \emph{\nd} hyperplane codes.
	\begin{defn} \label{D:nondegeneracy}
		An arrangement $(\cH,X)$ is \emph{\nd} if $X$ is open and convex, and the hyperplanes have \emph{\subnd\ intersections} in $X$, that is if $X \cap H_\s \od X \cap  \bigcap_{i\in\s} H_i  \neq \varnothing$, then $\dim H_\s = d - |\s|$.
		
		\noindent We call a code $\cC$ a \emph{\nd\ hyperplane code} if there exists a \nd\ arrangement $(\cH,X)$ such that $\cC = \code(\cH, X)$.  
	\end{defn}

\noindent 	\Nd\ arrangements are robust to noise in the sense that all atoms have nonzero measure. 
	\begin{lem}\label{L:nondegenerateatoms}
		If $(\cH,X)$ is a \nd\ arrangement, then every nonempty atom $A_\sigma^{\mathcal U}$ of the cover $\mathcal U=\left \{H_i^+\cap X\right\}$ has a nonempty interior.
	\end{lem}
	\begin{proof}
		Let $A_\s$ be a nonempty atom of the {\nd} arrangement $(\cH,X)$ and consider a point $x \in A_\s$.
		Let $\t = \{j \mid x \in H_j\}$ index the set of hyperplanes on which $x$ lies.
		Then $x$ has an open neighborhood $V$ inside $X \cap (\bigcap_{i\in\s}H_i^+) \cap (\bigcap_{j\not\in\s\cup\t}H_j^-)$.
		By \subndy, the set $\{w_i \mid i\in\t\}$ is linearly independent.
		Therefore, there exists some $v \in \R^d$ such that $w_i \cdot v < 0$ for all $i \in \t$.
		For sufficiently small $\varepsilon > 0$, $y = x + \varepsilon v \in V$; therefore for any $i\in\t,$
			\[ w_i\cdot y - h_i = w_i\cdot(x + \varepsilon v) - h_i = w_i \cdot \varepsilon v < 0, \]
		and thus $y \in X \cap (\bigcap_{i\in\s}H_i^+) \cap (\bigcap_{j\not\in\s}H_j^-)$, which is the interior of $A_\s$.
	\end{proof}
	
\begin{figure}[t]
	\begin{tabular}{c @{\hspace*{1cm}} c}
		(a) \raisebox{-0.5\height}{\includegraphics[height=4cm]{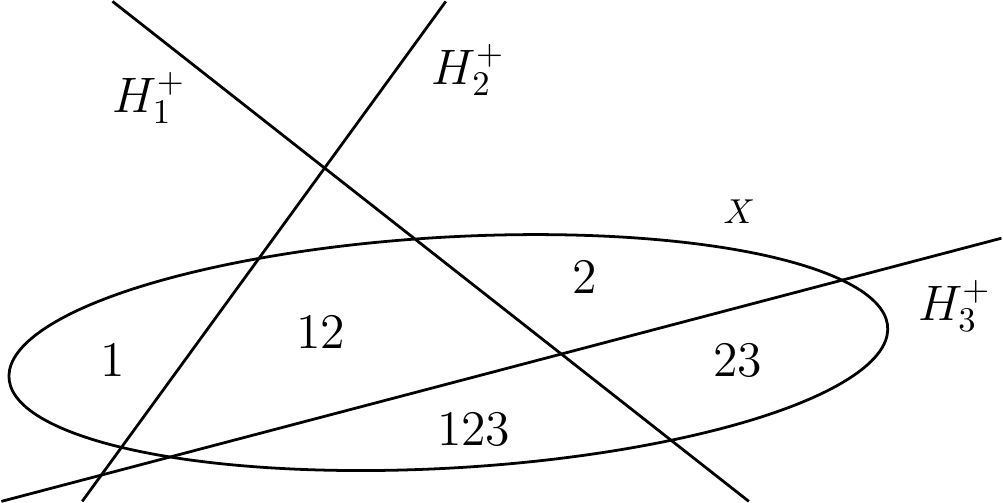}} &
		(b) \raisebox{-0.5\height}{\includegraphics[height=4cm,page=1]{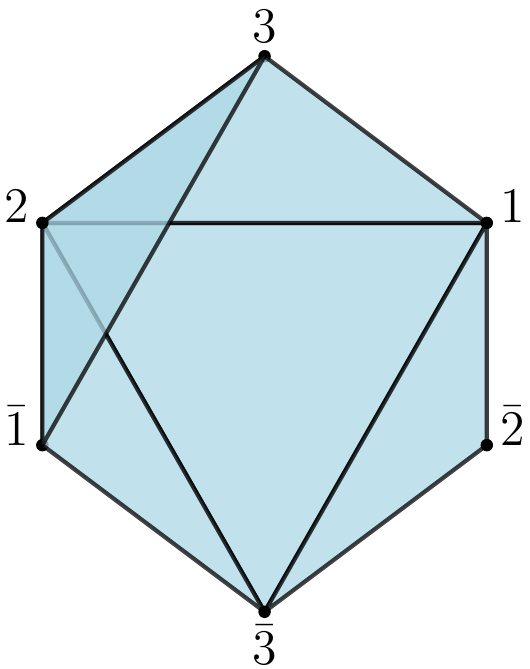}}
	\end{tabular}
	\caption{(a) \Nd\ arrangement $(\cH,X)$ with atoms labeled by their corresponding codewords. (b) The polar complex $\bd(\code(\cH,X))$, defined in \refS{polarcomplex}}
	\label{F:hyperplanecodeexample}
\end{figure}

\begin{ex}\label{Ex:code1}
	The code $\cC_1  = \{1, 12, 123, 2, 23\}$ is a \nd\ hyperplane code; a realization is illustrated in \refF{hyperplanecodeexample}(a). To avoid notational clutter, we adopt the convention of writing sets without brackets or commas, so the set $\{1,2\}$ is written $12$.
\end{ex}

\subsection{Bitflips and \nd~hyperplane codes.}  The abelian group  $(\mathbb Z _2 )^n$ acts on  $2^{[n]}$ by ``flipping bits'' of  codewords.  
Each generator $\e_i \in (\mathbb Z _2 )^n$ acts by flipping the
$i$-th bit, i.e. 
	\[ \e_i\cdot \sigma \od \begin{cases}
			\s \cup i &\text{if $  i \notin \s$} \\ 
			\sigma \setminus i &\text{if $  i \in \s$}. 
	\end{cases} \]
This  action extends to the action of $(\ZZ_2)^n$  on  codes, with 
$g \cdot\cC = \{ g\cdot \s \mid \s \in \C \}$.
The group $(\ZZ_2)^n$ also acts on oriented hyperplane arrangements. Here each generator $\e_i$ acts by reversing the orientation of the $i$-th hyperplane:
	\[ \e_i\cdot H_j ^+ \od \begin{cases}
							H_j ^+  &\text{if $  i \neq j $} \\ 
							H_j^-  &\text{if $  i=j$}. 
	\end{cases} \]
One might hope that applying bitflips commutes with taking the code of 
a hyperplane arrangement, but this is not true for arbitrary  hyperplane codes.

\begin{ex}\label{Ex:code2}
	Consider $ H_1^+,H_2^+,H_3^+\subseteq \mathbb R^2$, with $H_1^+= \{x+y>0\}$, $H_2^+= \{x - y>0\}$, and $H_3^+ = \{x>0\}$, illustrated in \refF{bitflipfailure}(a).
	By inspection, $\cC_2 = \code(\cH,\R^2)$ has  codewords $\{\varnothing, 1, 13, 123, 23, 2\}$. Meanwhile,
	\[ \code(\e_3\cdot \cH,\R^2) = \{3, 13,1,12,2,23,\varnothing\} = \e_3\cdot\code(\cH,\R^2)\cup \{ \varnothing\}. \] 
	The extra codeword appears because after flipping hyperplane $H_3$, 
	the origin no longer belongs to  the same atom as  the points to its left, 
	and thus produces	a new codeword, see \refF{bitflipfailure}(b).
\end{ex}

\begin{figure}[t]
		\begin{tabular}{cc @{\hspace*{1cm}} cc}
			(a) & \raisebox{-0.5\height}{\includegraphics[height=3cm,page=1]{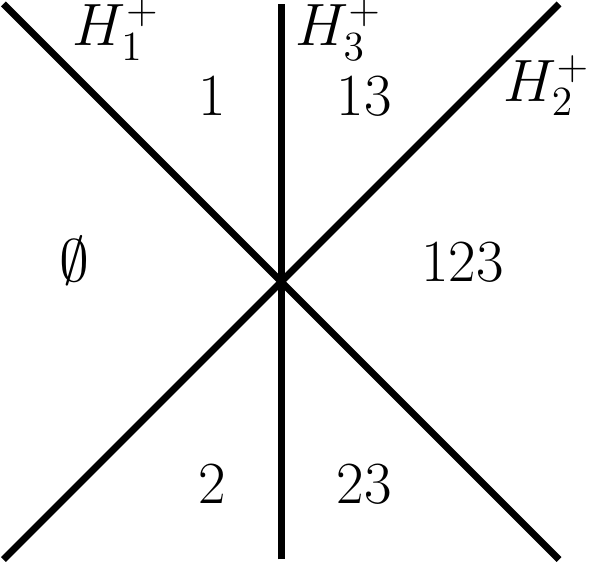}}
			&
			(b) & \raisebox{-0.5\height}{\includegraphics[height=3cm,page=2]{bitflipcounterexample.pdf}}\\
			& $(\cH,X)$ & & $(\e_3\cdot\cH,X)$\\
			\\
			(c) &  \raisebox{-0.5\height}{\includegraphics[height=3cm,page=2]{polarcomplex_examples.pdf}}
			& (d) & \raisebox{-0.5\height}{\includegraphics[height=3cm,page=3]{bitflipcounterexample.pdf}}
		\end{tabular}
		\caption{(a),(b)~The action of $(\ZZ_2)^n$ does not necessarily commute with taking the code of a non-\nd\ hyperplane arrangement. (c)~The polar complex $\bd(\code(\cH,X))$ for the arrangement in panel~(a) is an octahedron missing two opposite faces. (d)~A \nd\ realization of $\code(\mathbf{e}_3\cdot\cH,X)$, obtained from panel~(b) by translating $H_3$ to the left.}
		\label{F:bitflipfailure}
	\end{figure}

\noindent Nevertheless, the group action \emph{does} commute with taking the code of a \nd\ hyperplane arrangement.

	\begin{prop}\label{P:bitflipsinvariance}
		If $(\cH,X)$ is a \nd\ arrangement, then for every $g \in (\ZZ_2)^n$, $(g\cdot\cH,X)$ is also a \nd\ arrangement and
			\begin{align}
				\code(g\cdot\cH,X) = g\cdot\code(\cH,X). \label{eq:bitflipinvariance}
			\end{align}
	\end{prop}

	\begin{proof}
		Since the action of $(\ZZ_2)^n$ does not change the hyperplanes $H_i$ (only their orientation) nor the set $X$, the  \ndy\ is preserved. 
		By \refL{nondegenerateatoms}, each atom of $(\cH,X)$ has a nonempty interior; this interior is not changed by reorientation of the hyperplanes. 
		Thus, atoms are neither created nor destroyed by reorienting hyperplanes in a \nd\ arrangement; only their labels change,  and  $\code(g\cdot\cH,X) = g\cdot\code(\cH,X)$.
	\end{proof}

\subsection{The polar complex}\label{S:polarcomplex}
The invariance \eqref{eq:bitflipinvariance} of the class of {\nd}~hyperplane codes under the $(\ZZ_2)^n$ action 
 makes it natural to consider a simplicial complex whose structure is preserved by bitflips. The simplicial complex of the code is insufficient for this purpose: for any nontrivial code $\cC\subseteq2^{[n]}$ with a nonempty codeword, the simplicial complexes of the codes in the $(\ZZ_2)^n$-orbit of $\cC$ will include the full simplex on $n$ vertices, regardless of the structure of $\D(\cC)$. 
 
 We denote by $[n] \od \{1,\dots,n\}$ and $\br{[n]} \od \{\br{1},\dots,\br{n}\}$ two separate copies of the vertex set.  
Given a code $\cC \subseteq 2^{[n]}$,  define the \emph{polar complex},   $\Gamma(\cC)$, 	as a pure $(n-1)$-dimensional simplicial complex on vertex set  $[n] \sqcup \br{[n]}$ with facets in bijection with the codewords of $\cC$. 
\begin{defn}
	Let $\cC \subseteq 2^{[n]}$ be a combinatorial code.
	For every codeword $\sigma\in \C$ denote  
	\begin{align*}
		\Sigma(\sigma) &\od \{i \mid i\in\s\} \sqcup \{\bar{i} \mid i\not\in\s\} = \s \sqcup \br{[n]\sm\s}\\
		\intertext{and define the \emph{polar complex of $\cC$} as}
		\bd(\cC) &\od \Delta ( \{   \Sigma(\sigma)\mid \s \in  \cC \}).
	\end{align*}
\end{defn}

	Continuing \refEx{code1}, the polar complex of $\cC_1 = \{1,12,123,2,23\}$ is given by
		$ \bd(\cC_1) = \D(\{1\bar2\bar3, 12\bar3, 123, \bar12\bar3, \bar123\})$. 
It is depicted in \refF{hyperplanecodeexample}(b) as a subcomplex of the octahedron.
The polar complex $\bd(2^{[3]})$ consists of the eight boundary faces of the octahedron; generally, the polar complex of the code consisting of all $2^n$ codewords on $n$ vertices is the boundary of the $n$-dimensional {cross-polytope}.

The polar complex of code $\cC_2$ in Example \ref{Ex:code2} is depicted in \refF{bitflipfailure}(c). Note that it follows from \refPp{nondegenerateshelling} that $\cC_2$ is \emph{not} a \nd\ hyperplane code, due to the structure of $\bd(\cC_2)$. In contrast, while \refF{bitflipfailure}(b) depicts a non-\nd\ arrangement, the code of that arrangement has a \nd\ realization depicted in \refF{bitflipfailure}(d).
 
The action of the bitflips $(\ZZ_2)^n$ on the boolean lattice induces an action on the facets of the polar complex, so that 
$g\cdot  \Sigma( \sigma) = \Sigma(g\cdot \sigma) $. In particular,   $\Gamma(g\cdot \C) = g\cdot \Gamma( \C)$, and the complex $\Gamma(g\cdot \C) $ is isomorphic to $\Gamma(  \C) $.
The Stanley-Reisner ideal of $\bd(\cC)$ is closely related to the neural ideal, defined in \cite{neuralring};  this will
be elaborated in Section~\ref{S:algebra}.
Moreover, in the case of \nd\ hyperplane codes, $\bd(\cC)$ has a simple description as the nerve of a cover:
\begin{lem}\label{L:polarnerve}
	If $\cC = \code(\cH,X)$ is the code of a \nd\ hyperplane arrangement, then
		\begin{align}
			\bd(\cC) = \nerve(\{H_i^+\cap X, H_i^-\cap X\}_{i\in[n]}) \label{eq:polarnerve}
		\end{align}
\end{lem}
\begin{proof}
	Consider a maximal face $\S(\s) \in \bd(\cC)$. By \refL{nondegenerateatoms}, $A_\s$ has nonempty interior given by $X \cap \bigcap_{i\in\s} H_i^+ \cap \bigcap_{j\not\in\s} H_j^-$, hence $\S(\s) \in \nerve(\{H_i^+\cap X, H_i^-\cap X\}_{i\in[n]}).$ Likewise, if $F$ is maximal in the complex $\nerve(\{H_i^+\cap X, H_i^-\cap X\}_{i\in[n]})$, the subset consisting of unbarred vertices in $F$ is a codeword as the corresponding atom is nonempty.
\end{proof}


\section{Obstructions for hyperplane codes}\label{S:obstructions}

Here we describe several major \emph{hyperplane obstructions},
the  properties of a combinatorial code that are  necessary for it to be realized by a {\nd} hyperplane arrangement.

\subsection{Local obstructions and bitflips}\label{S:bitflips}

A larger class 
of codes that arises in the neuroscience context are the open convex codes \cite{neuralring,cruzetal,GI14, Carina2015}.   
 A code $\C\subset 2^{[n]}$ is called {\it open convex} 
if there exists  a collection $\cU$ of $n$ open and 
convex sets $U_i\subseteq X\subseteq \R^d$, such that 
 $\C=\code\left(\cU,X \right)$.
Not every combinatorial code is convex. One obstruction to being  
an open convex code stems from an analogue of the nerve lemma 
\cite{Bjorner:1996}, recently proved in \cite{shiu2018}; see also \cite{jeffsCUR}.

	Recall the \emph{link} of a face $\s$ in a simplicial complex $\D$ is the subcomplex defined by
		\[ \link_\s\D \od \{\nu\in\D \mid \s\cap\nu = \varnothing,\, \s\cup\nu \in \D\}. \]
	When $\s \not \in \code(\cU,X)$, yet $\s \in  \nerve(\cU)$, the subset $U_\sigma\od \bigcap_{i\in\s} U_i$ is covered by the collection  of  sets  $\left\{U_j \cap U_\sigma\right \}_{ j\not\in\s}$. It is easy to see that in this situation,  
	\begin{equation*}
	 \link_\s \nerve(\cU) =\nerve(\{U_j \cap U_\sigma \}_{ j\not\in\s}),
	 \end{equation*} 
	    see e.g. \cite{GI14,cruzetal,Carina2015}.

	\begin{defn}
		A pair of faces $(\s,\t)$ of a simplicial complex $\D$ is a \emph{free pair} if 
		$\t$ is a facet of $\D$,
		$\s\subsetneq\t$, and 
		$\s\not\subseteq\t'$ for any other  facet $\t'\neq\t$. 
		The simplicial complex
			\[ \del_\s \D \od \{ \nu\in\D \mid \nu \not\supseteq \s \}\]
		is called the \emph{collapse of $\D$ along $\s$,} and is denoted as $\D \searrow_{\s} \del_\s \D$. If a finite sequence of  collapses of $\D$ results in a new complex $\D'$, we write $\D \searrow \D'$. If $\D \searrow \{\}$, we say $\D$ is \emph{collapsible}.
	\end{defn}

	Note that  the \emph{irrelevant simplicial complex}  $\{\varnothing\}$,
	consisting of a single empty face,
	is \emph{not} collapsible, as there is no other face properly contained in $\varnothing$. However, the \emph{void complex} $\{\}$ with no faces is collapsible.
 
\begin{lem}[{\cite[Lemma 5.9]{shiu2018}},  \cite{jeffsCUR}]
\label{L:nervelemma} 
	For any collection $\mathcal U  = \{U_1,\dots,U_n\} $ of  
	open convex sets $U_i\subset \mathbb R^d$
	whose union $\bigcup_{i\in[n]}U_i$ is also convex,
	its nerve,  $\nerve(\mathcal U)$, 
	is collapsible. 
\end{lem}

	\begin{cor}[{\cite[Theorem 5.10]{shiu2018}}]\label{C:localobstruction}
		Let $\cC = \code(\cU,X)$ with each $U_i \subseteq X \subseteq \R^d$ open and convex. Then $\link_\s\D(\cC)$ is
		collapsible
		for every nonempty $\s\in\D(\cC)\sm\cC$.
	\end{cor}

	The last  observation provides a ``local obstruction'' for a code $\C$ being an open convex code: if a non-empty $\s \in \Delta(\cC) \setminus \cC$ has a non-collapsible link, then $\cC$ is nonconvex. It had been  previously known (see, for example, \cite[Theorem 3]{GI14}) that $\link_\s\D(\cC)$ is contractible under the hypotheses of \refC{localobstruction}. Since collapsibility implies contractibility but not vice versa, we refer to a face $\s\in\D(\cC)\sm\cC$ with non-collapsible link as a \emph{strong} local obstruction; if $\link_\s\D(\cC)$ is non-contractible, we refer to $\s$ as a \emph{weak} local obstruction.
	
	Half-spaces are convex, thus local obstructions to being a convex code  are also obstructions to being a 
	hyperplane code. Therefore \refP{bitflipsinvariance} implies a much stronger statement. Not only are local obstructions in $\cC$ forbidden, we must also exclude local obstructions in $g\cdot\cC$ for all bitflips $g \in (\ZZ_2)^n$, since $g\cdot\cC$ is also a stable hyperplane code. We make this precise below.

	\begin{defn}\label{D:bflo} Let $g \in (\ZZ_2)^n$ and $\t\subseteq[n]$ be a pair  such that  $\link_\t\Delta(g \cdot \cC)$ is not collapsible (respectively, contractible) and $  \t \notin   g \cdot  \cC$. Then $(g,\t)$ is called a \emph{strong} (resp.\ \emph{weak}) \emph{bitflip local obstruction.}
	\end{defn}

	\begin{bigthm}[Bitflip local property]\label{P:ndhcimpliesnobflo}
		Suppose $\cC$ is a \nd\ hyperplane code. Then $\cC$ has no strong bitflip local obstructions.
	\end{bigthm}
	\begin{proof} 
		Halfspaces are convex, thus $\cC$ has no strong local obstructions. By \refP{bitflipsinvariance}, $g\cdot \cC$ is a \nd\ hyperplane code for all $g\in(\ZZ_2)^n$. Hence, $g\cdot\cC$ has no strong local obstructions.
	\end{proof}
	
	The nomenclature of ``weak'' and ``strong'' local obstructions signifies that a code with no strong local obstructions has no weak local obstructions, but generally not vice-versa. In particular, a stable hyperplane code also has no weak bitflip local obstructions.

\begin{figure}[t]
	\includegraphics[height=3.5cm]{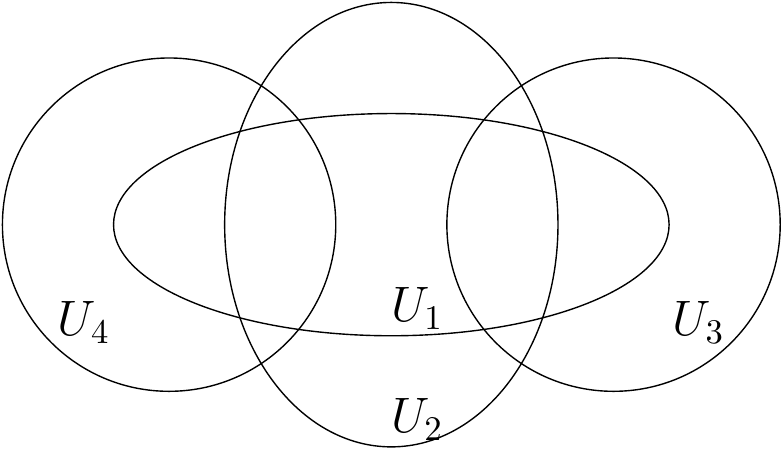}
	\caption{An open convex realization of $\cC_3$ with $X = \R^2$.}
	\label{F:convexnotndhc}\label{Fig:code3}
\end{figure}

\begin{ex}\label{Ex:code3}
	The code $\cC_3 = \{\varnothing,2,3,4,12,13,14, 23,24,123,124 \}$
	is realizable by open  convex sets in $\mathbb{R}^2$ (see \refF{convexnotndhc}), 
	and thus it cannot have local obstructions to convexity. Flipping bit 2 yields
	\[ \e_2 \cdot \cC_3 = \{2,\varnothing,23,24,1,123,124,3,4,13,14\}. \]
	The new simplicial complex $\Delta(\e_2\cdot\cC_3)$ has facets $123$ and $124$.
	The edge $12$ is \emph{not} in the code and $\link_{12}\D(\e_2\cdot \cC_3)$ is two vertices;
	therefore, $(\e_2,12)$ is a bitflip local obstruction and $\cC_3$ is not a {\nd} hyperplane code.
\end{ex}

	It is worth highlighting an essential feature of the polar complex that makes it a natural  tool for studying hyperplane codes, in light of the bitflip local property. For every $g\in(\ZZ_2)^n$, the simplicial complex $\D(g\cdot\cC)$ is isomorphic to an induced subcomplex of $\bd(\cC)$: Let $\s$ denote the support of $g$ and define
	\[ \bd(\cC)|_{([n]\sm\s) \sqcup \br{\s}} \od \{F \in \bd(\cC) \mid F \subseteq ([n]\sm\s) \sqcup \br{\s}\}. \]
	Then $\bd(\cC)|_{[n]\sm\s \sqcup \br{\s}} \cong \D(g\cdot\cC)$, with the isomorphism given by ``ignoring the bars,'' i.e.\ $i \mapsto i$ for $i \in[n]\sm\s$ and $\br{j}\mapsto j$ for $j \in\s$. Thus we can find bitflip local obstructions directly in the polar complex as follows.

\begin{prop}\label{P:bfloinpolarcomplex}
	Let $\cC\subseteq 2^{[n]}$ be a code, $g \in (\ZZ_2)^n$ with $\s$ its support, and let $\t \subseteq [n]$. Then $(g,\t)$ is a bitflip local obstruction for $\cC$ if and only if
		\[ g\cdot\t \sqcup \br{[n]\sm g\cdot\t} \not\in \bd(\cC)
			\qquad \text{and} \qquad
			\link_{g\cdot\t}\bd(\cC)|_{([n]\sm\s) \sqcup \br{\s}} \text{ is not collapsible}. \]
\end{prop}

\begin{proof}
		Note that $g\cdot\t \sqcup \br{[n]\sm g\cdot\t} \not\in \bd(\cC)$ if and only if $\t\not\in g\cdot\cC$. The complex $\bd(\cC)|_{([n]\sm\s)\sqcup\br{\s}}$ is isomorphic to $\D(g\cdot\cC)$, and
		\[ \link_\t\D(g\cdot\cC) \cong \link_{g\cdot\t}(\bd(\cC)|_{([n]\sm\s)\sqcup\br{\s}}).\]
		Hence, the conditions of the proposition are equivalent to the conditions of \refD{bflo}.
	\end{proof}

	\bigskip 

\subsection{Spherical Link Obstructions}\label{S:spherelinkobstructions}

Here we introduce another  obstruction  that can be detected via the polar complex
of {\nd} hyperplane codes. We use the following notation to aid our discussion. For a face $F \in\bd(\cC)$, we write $F = F^+ \sqcup \br{F^-}$ to denote the restrictions of $F$  to $[n]$ and $\br{[n]}$. The \emph{support} of $F$ is $\underline{F} = F^+ \cup F^-$, the set of (barred or unbarred) vertices appearing in it. 

For \nd\ arrangements $(\cH,X)$, Lemmas~\ref{L:nondegenerateatoms} and \ref{L:polarnerve} allow us to translate between faces of $\bd(\code(\cH,X))$  and convex subsets of $X$ as follows: The face $F = F^+\sqcup\br{F^-} \in\bd(\cC)$ corresponds to the open convex set
		\[ R_F = X \cap \Bigl(\bigcap_{i\in F^+} H_i^+\Bigr) \cap \Bigl(\bigcap_{{j}\in\br{F^-}} H_j^-\Bigr).\]
Note that for a facet  $F = \s\sqcup\br{([n]\sm\s)}$ of the polar complex,  $R_F$ is precisely the interior of the atom $A_\s$.
In addition, it is easy to see  that 
$\link_F\bd(\cC) = \bd(\cC')$ for some 
$\cC' \subseteq 2^{[n]\sm\underline{F}}$. Therefore, 
we consider the topology of the covered subset of 
$R_F$.
We show the positive and negative halfspaces indexed by the complement of $\underline{F}$ will cover either all of $R_F$ or all but a linear subspace of $R_F$. The following proposition  describes the combinatorics of the nerve of this cover.
\begin{prop}\label{Prop:37}  
  Let $(\cH,X)$ be a stable arrangement, and let $R_F$ be a nonempty region with $|F| < n$.
  Then $(\{H_i^+\cap R_F\}_{i\not\in\underline{F}}, R_F)$ is a stable arrangement.
  Moreover,  the $\nerve(\{H_i^+,H_i^-\}_{i\not\in\underline{F}})$ is either collapsible
  or is the polar complex of the full code on the vertices  $[n]\sm\underline{F}$,
  i.e. $\nerve\left(\{H_i^+,H_i^-\}_{i\not\in\underline{F}}\right )=
  \Gamma\left(2^{[n]\sm\underline{F}}\right )$. 
\end{prop}

\begin{proof} Denote  $\nu \od  [n] \sm \underline{F}$.  
 	First we verify the arrangement $(\{H_i^+\cap R_F\}_{i \in \nu }, R_F)$
	is \nd. The region  $R_F$ is open and convex, and intersections of hyperplanes in $R_F$ lie in $X$, so they already satisfied the \subndy\ condition.

	Consider $H_\nu \cap R_F$;  if it is empty, then the union of the positive and negative open half-spaces indexed by $\nu$ is all of the convex set $R_F$, and so by \refL{nervelemma} the nerve is collapsible. If $H_\nu \cap R_F \neq \varnothing$, by \ndy, we have $\dim H_\nu = d - |\nu|$. In this case, the linear independence of $\{w_i \mid i\in\nu\}$ ensures all of the $2^{|\nu|}$ intersection patterns of halfspaces, i.e.\ the nerve is $\bd(2^\nu) = \bd(2^{[n]\sm\underline{F}})$.
\end{proof}

\begin{defn}\label{D:slo}
	Let $F\in\bd(\cC)$ be a non-maximal face such that $\link_F(\bd(\cC))$ is neither collapsible nor $\link_F(\bd(\cC)) = \bd(2^{[n]\sm\underline{F}})$.
	We call $F$ a \emph{sphere link obstruction}.
\end{defn}

\medskip 

By \refL{polarnerve}, we have $\link_F\bd(\cC) = \nerve(\{H_i^+\cap R_F, H_i^-\cap R_F\}_{i \not\in\underline{F}})$. This, together with the Proposition  \ref{Prop:37}, imply 

\begin{bigthm}[Sphere link property]\label{P:ndhcimpliesnoslo}
	Suppose $\cC$ is a {\nd} hyperplane code. Then $\cC$ has no sphere link obstructions.
\end{bigthm}

\begin{ex}
	Continuing \refEx{code2}, we consider the polar complex $\bd(\cC_2)$ for the un\nd\ arrangement $(\cH,X)$ in \refF{bitflipfailure}(a). This complex is illustrated in \refF{bitflipfailure}(c). The face $\varnothing$ is a sphere link obstruction: $\link_\varnothing\bd(\cC_2) = \bd(\cC_2),$ and this complex is neither the complex $\bd(2^{[3]})$, which would have 8 facets, nor is it collapsible. Therefore, $\cC_2$ is not a \nd\ hyperplane code.
\end{ex}

\bigskip 
\subsection{Chamber Obstructions}

The intuition  behind the third  obstruction in this section
concerns maximal hyperplane intersections. If a collection $\{H_i\}_{i\in\s}$
of hyperplanes intersects in a point ($\dim H_\s = 0$), 
then that point has fixed position relative to other hyperplanes. 
In particular, there cannot be two distinct regions defined by the other hyperplanes that contain that point. 
More generally, if $H_\sigma \neq \varnothing$ is a maximal non-empty intersection, then it intersects only one atom of the arrangement  $\{ H_j\}_{j\not\in \sigma}$ of the remaining hyperplanes.

\begin{defn}\label{D:chamber}
	The \emph{geometric chamber complex} of a hyperplane arrangement $\cH$ relative to an open convex set $X$, $\cham(\cH,X)$, is the set of $\s\subseteq[n]$ such that $H_\s \cap X \neq \varnothing$. By convention, $H_\varnothing = \R^d$ so $\varnothing \in \cham(\cH,X)$ for all $(\cH,X)$.
	
	\noindent The \emph{combinatorial chamber complex} of a code $\cC$, denoted $\cham(\cC)$, is given by the set of $\s \subseteq [n]$ such that there exists $T \in \bd(\cC)$ with $\underline{T} = [n] \setminus \s$ and $\link_T\bd(\cC) = \bd(2^\s)$. We call such a subset $T$ a \emph{chamber} of $\s$. 
\end{defn}

Both $\cham(\cH,X)$ and $\cham(\cC)$ are simplicial complexes: the former because for any $i\in \sigma$,  $H_{\s\sm i}\supseteq H_\s$; the latter because if $\link_T\bd(\cC) = \bd(2^\s)$ then $\link_{T\cup i} = \bd(2^{\s\sm i})$. 
For \nd\ hyperplane codes, the facets of these simplicial complexes correspond to maximal hyperplane intersections.

\begin{ex}
	Returning to the {\nd} code $\cC_1$ from \refEx{code1}, the maximal faces of $\cham(\cC_1)$ are $2$ and $13$. This is because $\link_{1\bar{3}}(\bd(\cC_1)) = \bd(2^{\{2\}})$ and $\link_{2}(\bd(\cC_1)) = \bd(2^{\{1,3\}})$.
	By inspection, these are also maximal faces of the geometric chamber complex $\cham(\cH,X)$ for the arrangement in \refF{hyperplanecodeexample}(a).
\end{ex}

\begin{prop}\label{P:hyperplanesinglechamber}
	For a \nd\ arrangement $(\cH,X)$, the associated chamber complexes coincide, $\cham(\cH,X) = \cham(\code(\cH,X))$. Moreover, for $\cC = \code(\cH,X)$, each facet $\s$ of $\cham(\cC)$ has a \emph{unique} chamber $T \in \bd(\cC)$.
\end{prop}

	\begin{proof}
		Let $(\cH,X)$ be a \nd\ pair and set $\cC = \code(\cH,X)$. Suppose $\s \in \cham(\cH,X)$, so $H_\s \cap X \neq \varnothing$. Then, for any atom $A_\t$ of the arrangement $(\{H_i^+\cap X\}_{i\not\in\s},X)$ such that $H_\s \cap A_\t \neq \varnothing$, the set $T = \t \sqcup \br{([n]\sm\s)\sm\t}$ is a chamber of $\s$, hence $\s\in\cham(\cC)$.
		For the reverse containment, suppose $\s\in\cham(\cC)$ has chamber $T$. Then
			\[ \bd(2^\s) = \link_T\bd(\cC) = \bd(\code(\{H_i^+\cap R_T\}_{i\not\in\s}, R_T)), \]
		meaning the hyperplanes $\{H_i\}_{i\in\s}$ partition $R_T$ into the maximal number of regions, i.e.\ it is a central arrangement. Thus $H_\s \cap R_T \neq \varnothing$ and therefore $H_\s \cap X \neq \varnothing$ and $\s \in \cham(\cH,X)$.

		Now consider $\s$ a facet of $\cham(\cC)$. Because $\cC = \code(\cH,X)$,  the intersection of hyperplanes $H_\s\cap X$ does not meet any other hyperplanes inside  $X$. Therefore, it is interior to only one atom of the arrangement $(\{H_j^+\}_{j\not\in\s},X)$; the face in $\bd(\cC)$ corresponding to this atom is the unique chamber $T$.
	\end{proof}

\noindent We reformulate \refP{hyperplanesinglechamber} into our third and final obstruction to hyperplane codes.

\begin{defn}\label{D:sco}
	Let $\s\subseteq[n]$ be a maximal face of $\cham(\cC)$ such that there exist two faces $T_1\neq T_2 \in \bd(\cC)$ with $\link_{T_1}\bd(\cC) = \link_{T_2}\bd(\cC) = \bd(2^\s)$. Then we call $\s$ a \emph{chamber obstruction}.
\end{defn}

\begin{bigthm}[Single chamber property]\label{P:ndhcimpliesnoco}
	Suppose $\cC = \code(\cH,X)$ is a \nd\ hyperplane code. Then $\cC$ has no chamber obstructions.
\end{bigthm}

\begin{ex}
	The code $\cC_3$ from \refEx{code3} also has a chamber obstruction, in the form of $\s = \{1,2\}$. 
	There are two faces $\{\bar{3},4\}$
	and $\{3,\bar{4}\}$ with link in $\bd(\cC_3)$ equal to the full polar complex on
	$\{1,2\}$. One can check that this is maximal in $\cham(\cC_3)$, 
	creating a chamber obstruction.
\end{ex}


\section{The main results}\label{S:mainresults}

Our main results consist of  showing that (i) the polar complex of  a \nd\ hyperplane code is shellable and 
(ii) shellability of $\bd(\cC)$ implies $\cC$ has none of the obstructions thus far considered, except possibly the strong bitflip obstruction. 
	First, we define shellability.

\begin{defn}\label{D:shellabilityeasy} 
 Let $\D$ be a pure simplicial complex of dimension $d$ and $F_1,\dots,F_t$ an ordering of its facets. The ordering is a \emph{shelling order} if, for $i > 1$, the complex
	  	\[ \D(\{F_i\}) \cap \D(\{F_1,\dots,F_{i-1}\}) \]
is pure of dimension $d-1$. A simplicial complex is \emph{shellable} if its facets permit a shelling order.
\end{defn}

A shelling order constructs a simplicial complex one facet at a time in such a way that each new facet is glued along maximal faces of its boundary. The facets of $\bd(\cC)$ correspond to codewords of $\cC$, thus a shelling order of $\bd(\cC)$ corresponds to an ordering of the codewords. We explicitly construct such an order in \refS{proofs:shellability} to prove \refPp{nondegenerateshelling}. 

\begin{bigthm}\label{P:nondegenerateshelling}
	Let $\cC \subseteq 2^{[n]}$ be a \nd\ hyperplane code. Then $\bd(\cC)$ is shellable.
\end{bigthm}

	It turns out that the structure of shellable polar complexes does not allow for many of the obstructions thus far considered.
 
	\begin{bigthm}\label{T:shellableimplies}
		Let $\cC \subseteq 2^{[n]}$ be a combinatorial code such that $\bd(\cC)$ is shellable. Then,
		\vspace{-0.5\baselineskip}
		\begin{enumerate}[label={\arabic*.},leftmargin=*]
			\item\label{T:shellableimpliesnobflo} $\cC$ has no \emph{weak} bitflip local obstructions,
			\item\label{T:shellableimpliesnoslo} $\cC$ has no sphere link obstructions, and
			\item\label{T:shellableimpliesnoco} $\cC$ has no chamber obstructions.
		\end{enumerate}
	\end{bigthm}
	
	\refT{shellableimplies} is proven in \refS{proofs:shellabilityobstructions}. Note the conclusion of \refT{shellableimplies}.1 refers to weak local obstructions, highlighting the gap between the notions of collapsibility and contractibilty.


\section{Discussion}\label{S:discussion}
Hyperplane codes are a special class of convex codes that naturally arise as the output of a one-layer feedforward network \cite{GI14}.   Hyperplane codes are a proper\footnote{See e.g. Example \ref{Ex:code3} and Figure \ref{Fig:code3}.} subclass of  the open convex codes.   We set out to  find obstructions to being a hyperplane code, while focusing on {\nd} hyperplane codes.
There are two   reasons for  primarily considering  the {\nd} hyperplane codes:  (i)  they are `generic' in that  they  are stable to small perturbations, and  (ii) they allow  the 
action of the group of bitflips $(\mathbb Z_2)^n$.  The second  property makes it natural to consider the  polar complex $\Gamma(\C)$ of a code, because the combinatorics of the polar complex captures all the bitflip-invariant properties of the underlying {\nd} hyperplane code. 
We have established  the following relationships among the properties of the polar complex of the code. The necessary conditions for $\C$ being a stable hyperplane code,
\newcommand\follows{\Longleftarrow}
$$
	  \Gamma(C) \text{ is shellable}  \follows \, \, 
	  \begin{matrix} \cC \text{ is a \nd} \\ \text{ hyperplane code } \end{matrix}
	   \implies \, \, 
	  \begin{cases} \cC \text{ has no {\it strong} bitflip obstructions,}   \\ \cC  \text{ has no sphere link obstructions,} \\ \cC \text{ has no chamber obstructions.} \end{cases}
$$
We have also established that almost all currently known necessary conditions follow from the shellability of the polar complex: 
$$
\Gamma(C) \text{ is shellable}  \implies \,  \begin{cases} \cC \text{ has no {\it weak}  bitflip obstructions,}   \\ \cC  \text{ has no sphere link obstructions,} \\ \cC \text{ has no chamber obstructions.} \end{cases} 
$$
 Note  that the shellability of the polar complex implies the lack of {\it weak}  bitflip obstructions, while a {\nd} hyperplane code lacks {\it strong} bitflip obstructions. It is currently an open problem if the gap between the strong and the weak versions of the local obstructions is indeed a property of  shellable polar complexes. Alternatively, codes with shellable polar complexes may also lack the strong bitflip obstructions.
An example of a code whose polar complex is shellable, but  has  the strong bitflip obstruction\footnote{In particular, the appropriate link in Definition \ref{D:bflo}  is contractible, but not collapsible.} 
 would provide a negative answer to the following open question:  Is  shellability of the polar complex equivalent to the code being a \nd\ hyperplane code?  
	  
What makes a code a {\nd}  hyperplane code is still an open question. It seems likely that the shellability of the polar complex is not {\it the only} necessary condition for a code to be a {\nd} hyperplane code. 
From a computational perspective, deciding if a given pure simplicial complex is shellable is known to be an NP-hard  problem \cite{Shellability_is_NP_Complete}. This likely means that answering  the question of whether a given code is produced by a one-layer network may be not computationally feasible. Ruling out that a given code is a hyperplane code may be less computationally intensive however, as it can rely on computing the Betti numbers of the free resolution of the Stanley-Reisner ideal of the polar complex, as illustrated in the following section.


\section{Algebraic signatures of a hyperplane code}\label{S:algebra}

Given a code $\cC$, how can we rule out that $\C$ is a {\nd} hyperplane code? In this section,
we show how the  tools from computational commutative algebra can be used to detect sphere link
obstructions via Stanley-Reisner theory.

\subsection{The neural and the Stanley-Reisner ideal}

The connections between neural codes and Stanley-Reisner theory were 
first developed in \cite{neuralring}, and later expanded upon in
\cite{homomorphisms}, \cite{garcia}, and \cite{Jeffries}. 
The key observation is that
a code $\C \subseteq 2^{[n]}$ can be considered as a set of points in $(\mathbb{F}_2)^n$, 
and the vanishing ideal $I_\cC$ of that variety is a ``pseudo-monomial ideal''
with many similarities to a monomial ideal. In this section,
we show that this connection can be made more explicit via the polar
complex.

First, we state necessary  prerequisites about the neural ring.
Let $ \mathbb F_2$ denote  the field with two elements, and consider the polynomial ring $ R \od \mathbb F_2 [x_1,\dots, x_n]$. 
A polynomial $f\in R$ can be considered as a function $f:2^{[n]}\to \mathbb F_2$ by defining $f(\s)$ as the evaluation of $f$ with $x_i = 1$ for $i \in \s$ and $x_i = 0$ for $i \not\in\s$.
Polynomials of the form 
	\[x^{\sigma} (1-x)^{\tau}\, \od \prod_{i \in \sigma } x_i \prod_{j \in \tau}(1-x_j),\]
where $\sigma, \tau \subseteq[n]$,  are said to be {\em pseudo-monomials}.
Note that the pseudo-monomial $x^\s(1-x)^{[n]\sm\s}$ evaluates to 1 if and only if the support of $x$ equals $\s$; such a pseudo-monomial is called the \emph{indicator function} of $\s$.

\begin{defn}[\cite{neuralring}]
	The \emph{vanishing ideal} of a code $\cC\subseteq2^{[n]}$ is the ideal of polynomials that vanish on all codewords of $\cC$,
	\begin{align*}
		I_\cC & \od \{f \in R \mid f(\s) = 0 \text{ for all } \s\in\cC\}. 
	\intertext{The \emph{neural ideal} of $\cC$ is the ideal generated by indicator functions of non-codewords,}
		J_\cC & \od   \left\langle  x^{\sigma} (1-x)^{[n]\setminus \sigma }\mid \sigma \notin \cC \right\rangle. 
	\intertext{The \emph{boolean ideal} of $\cC$ is the ideal generated by the \emph{boolean relations}, pseudo-monomials with $\s = \t = i$,}
		\mathcal{B} & \od \left \langle x_i (1-x_i) \mid i \in [n] \right \rangle.
	\end{align*}
\end{defn}

\begin{lem}[{\cite[Lemma 3.2]{neuralring}}]
	Let $\cC$ be a neural code. Then $I_\cC = J_\cC + \cB$.
\end{lem}

\noindent Pseudomonomials in the vanishing ideal $I_\cC$ correspond to relations of the form $\bigcap_{i\in\s} U_i \subseteq \bigcup_{j\in\t} U_j$ among sets in any cover realizing $\cC$.

\begin{lem}[{\cite[Lemma 4.2]{neuralring}}]\label{L:ICmembership}
	Let $\cC = \code(\cU,X)$ be a combinatorial code. Then
		\[ x^\s (1-x)^\t \in I_\cC \iff \bigcap_{i\in\s}U_i \subseteq \bigcup_{j\in\t}U_j, \]
	where by convention $\bigcap_{i\in\varnothing} U_i = X$ and $\bigcup_{j\in\varnothing} U_j = \varnothing$.
\end{lem}

\noindent In particular, the generators of $\cB$ correspond to the {tautological relations} $U_i \subseteq U_i$. The neural ideal records the non-tautological relations.

\begin{defn}[\cite{neuralring}]
	A {pseudo-monomial} $f \in J_{\cC}$ is said to be \emph{minimal} if there is no other pseudo-monomial $g \in J_\cC$  that   divides $f$.
	The \emph{canonical form} of $J_\cC$, denoted $CF(J_\cC)$, is the set of all the minimal pseudo-monomials in $J_\cC$.
  \end{defn}

The elements of the canonical form correspond to the minimal nontrivial relations $\bigcap_{i\in\s} U_i \subseteq \bigcup_{j\in\t} U_j$.
We will see that the canonical form of $J_\cC$ and the Boolean relations also corresponds with the generating set of the Stanley-Reisner ideal of $\bd(\cC)$.
We make these relationships explicit in \refL{canonical} and \refC{dualrelations}.

The Stanley-Reisner correspondence associates to any simplicial complex on $n$ vertices an ideal generated by square-free monomials in a polynomial ring in $n$ variables \cite{stanleyCCA}.
The construction of the polar complex is seen to be particularly natural when considering its associated Stanley-Reisner ideal.
For the unbarred vertices, we set the corresponding variables via $i \mapsto x_i$; for the barred vertices, we associate $\bar{i} \mapsto y_i$.
The Stanley-Reisner ideal of $\bd(\cC)$ is the ideal in $S \od \mathbb F_2 [x_1,\dots,x_n,y_1,\dots,y_n]$ generated by the squarefree monomials indexed by \emph{non}-faces of $\bd(\cC)$. 

\begin{defn}
	Let $\cC\subseteq 2^{[n]}$ be a combinatorial code. The Stanley-Reisner ideal of the polar complex is given by
		\[ I_{\bd(\cC)} = \langle x^\s y^\t \mid \s\sqcup\br{\t}\not\in\bd(\cC) \rangle \subseteq S. \]
\end{defn}

\begin{ex}\label{ex:stanleyreisner}
	Consider the code $\cC_1 = \{1, 12, 123, 2, 23\}$ from \refEx{code1}. 
	The corresponding variety in $\mathbb{F}_2^3$ is $\{100, 110, 111, 010, 011\}$ with canonical form given by
		\[ CF(J_{\cC_1}) = \{(1-x_1)(1-x_2),\; x_3(1-x_2)\}. \]
	The polar complex of $\cC_1$ is given by
		\[ \bd(\cC_1) = \D(\{ 1\bar2\bar3,\, 12\bar3,\, 123,\, \bar12\bar3,\, \bar123 \}). \]
	The minimal nonfaces of $\bd(\cC_1)$ are $\{1\bar1,\, 2\bar2,\, 3\bar3,\, \bar1\bar2,\, \bar23\}$. This gives the Stanley-Reisner ideal
		\[ I_{\bd(\cC_1)} = \langle x_1y_1,\, x_2y_2,\, x_3y_3,\, y_1y_2,\, x_3y_2\rangle. \]
	The first three monomials in this list correspond to the Boolean relations, while the last two can be compared to the canonical form.
\end{ex}

\noindent The intuition intimated by Example~\ref{ex:stanleyreisner} holds true in general. 

\begin{lem}\label{L:canonical}
	For any nonempty combinatorial code $\cC \subseteq 2^{[n]}$, the Stanley-Reisner ideal of the polar complex is induced by the canonical form and the Boolean relations.
	That is,
		\begin{align}
			x^\s y^\t \in I_{\bd(\cC)} \iff x^\s(1-x)^\t \in I_\cC.\label{eq:IC-IGC}
		\end{align}
	and so
		\begin{align}
			I_{\bd(\cC)} = \langle\, x^\sigma y^\tau \mid x^\s(1-x)^\t \in CF(J_\cC)\,\rangle + \langle\, x_i y_i \mid i \in [n] \,\rangle. \label{eq:IgammaC}
		\end{align}
\end{lem}

\begin{proof}[Proof of \refL{canonical}]
	Consider a square-free monomial $x^\s y^\t \in S$.
	By definition, $x^\s y^\t \in I_{\bd(\cC)}$ if and only if $\s\sqcup\br{\t}$ is a nonface of $\bd(\cC)$.
	The set $\s\sqcup\br{\t}$ is a nonface of $\bd(\cC)$ if and only if any codeword in $\cC$ which contains $\s$ is not disjoint from $\t$, that is, $\cC$ satisfies the following property:
		\begin{align}
			\text{for all $\a\in\cC$,} \qquad \s\subseteq \a \quad \implies \quad \a \cap \t \neq\varnothing.\label{eq:Cproperty}
		\end{align}
	If $\cC$ satisfies \eqref{eq:Cproperty}, the pseudomonomial $x^\s(1-x)^\t$ vanishes on all of $\cC$, as $x^\s$ evaluates to 0 on any codeword not containing $\s$, and $(1-x)^\t$ evaluates to 0 on any codeword not disjoint from $\t$, e.g.\ any codeword containing $\s$.
	Conversely, if $x^\s(1-x)^\t$ vanishes on all of $\cC$, every codeword that contains $\s$ must not be disjoint from $\t$, so $\cC$ satisfies \eqref{eq:Cproperty}.
	Therefore, $x^\s(1-x)^\t \in I_\cC$.
	Thus we have established \eqref{eq:IC-IGC} and \eqref{eq:IgammaC} follows, as any pseudomonomial in $I_\cC$ is divisible either by $x_i(1-x_i)$ for some $i$, or by an element of the canonical form $CF(J_\cC)$.
\end{proof}

\noindent The following is an immediate corollary of \refL{ICmembership} and \refL{canonical}.

\begin{cor}\label{C:dualrelations}
	Let $\cC = \code(\cU,X) \subseteq 2^{[n]}$ and $I_{\bd(\cC)}$ the Stanley-Reisner ideal of the polar complex of $\cC$. Then
		\[ x^\s y^\t \in I_{\bd(\cC)} \iff \bigcap_{i\in\s} U_i \subseteq \bigcup_{j\in\t} U_j. \]
\end{cor}

\subsection{Sphere link obstructions and multigraded free resolutions}

In \refS{spherelinkobstructions}, we showed that
$\link_{\S}(\bd(\cC))$ is either empty,
collapsible,
or is isomorphic to a sphere of 
dimension $n-|\S|-1$ when $\cC$ is a \nd\ hyperplane code. One consequence of this fact is that if
a \nd\ hyperplane realization of $\cC$ exists, then
a lower bound on the dimension of the realizing space is
\[ d \geq \max_{\S\in \bd(\cC)}\Bigl\{ (n- |\S|) \mid \link_{\S}(\bd(\cC)) \sim S^{n-|\S|-1}\Bigr\}.\]
However, this may not be the true lower bound.

\begin{figure}[h!]
	\begin{tabular}{c @{\hspace*{1cm}} c}
		(a) \raisebox{-0.5\height}{\includegraphics[height=3cm]{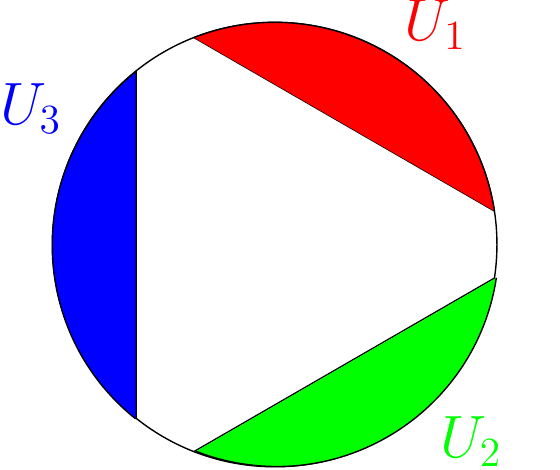}} &
		(b) \raisebox{-0.5\height}{\includegraphics[height=3cm,page=8]{polarcomplex_examples.pdf}}
	\end{tabular}
	\caption{(a)~Realization of $\cC_4 = \{\varnothing, 1,2,3\}$ in $\mathbb{R}^2$.
	Though sphere link dimension is 1, minimal realization dimension is 2. (b)~The polar complex $\bd(\cC_4)$. The only non-collapsible links are of the form $\link_{\br{i} \br{j}}\bd(\cC)$ }
	\label{fig:dimEx}
\end{figure}

\begin{ex}
Consider the code $\cC = \{\varnothing, 1,2,3\}$ consisting of 
four words; this can be realized by hyperplanes in $\R^2$ as in Figure~\ref{fig:dimEx}.
Still, the polar complex $\bd(\cC)$ 
has facets $\overline{123},1\overline{23},\bar{1}2\bar{3},\overline{12}3$,
which has spherical links only at $\S = \{\bar{i},\bar{j}\}$ for $i \neq j \in
\{1,2,3\}$. This might lead us to infer that the minimal realizing dimension
is $n - |\S|=  3 - 2 = 1$; however, it is easy
to prove that it is impossible to realize
by hyperplanes in $\mathbb{R}^1$.

\end{ex}

Another consequence of the sphere link property (\refPp{ndhcimpliesnoslo}) relates to algebraic properties of the Stanley-Reisner ring. 
The dual version of Hochster's formula relates the multigraded minimal free resolution of the Stanley-Reisner ideal to the simplicial homology of the corresponding complex.
A full exposition of minimal free resolutions is beyond the scope of this article, so we give a brief description and direct the reader to \cite[Chapter 1]{CCA} for more information. 

The \emph{multidegree} of a monomial $\left(\prod_{i=1}^n x_i^{a_i}\prod_{j=1}^n y_j^{b_j}\right)\in S$ is the vector of exponents $(a,b) = (a_1,\dots,a_n,b_1,\dots,b_n) \in \mathbb{N}^{2n}$. 
When the exponents are all 0 or 1, we identify the the multidegree with its support as a subset of $[n]\sqcup\br{[n]}$.
The \emph{coarse degree} of a monomial is the sum of the exponents $\sum_{i=1}^n a_i + \sum_{j=1}^n b_j \in \mathbb{N}$.
For a homogeneous ideal $I\subset S$, a \emph{minimal free resolution} of $S/I$ is an exact sequence of free modules that terminates in $S/I \to 0$.
Each module in the minimal free resolution of $S/I$
can be \emph{multigraded} so that each map in the resolution
preserves multidegree.
The \emph{multigraded Betti number} of $S/I$, $\beta_{i,\sigma} = \beta_{i,\sigma}(S/I)$, is the rank of the free module in position $i$ in the free resolution and with multidegree $\sigma$.
Importantly for our purposes, these Betti numbers can be explicitly computed with {\tt Macaulay2} \cite{M2} and similar computational algebra software.

\begin{lem}[Hochster's formula, dual version {\cite[Corollary 1.40]{CCA}}]\label{hochster}
	For $\bd(\cC)$ the polar complex of a code $\cC\subseteq 2^{[n]}$ and $\S$ a face of $\bd(\cC)$,
	\[ \beta_{i+1,\S^c}(S/I_{\bd(\cC)^\vee}) = \dim_k \tilde{H}_{i - 1} (\link_{\S}{\bd(\cC)} ; k). \]
	Here $\S^c = ([n]\sqcup\br{[n]})\sm\S$ denotes the complement of $\S$ in the vertex set of $\bd(\cC)$, and $\bd(\cC)^\vee$ denotes the Alexander dual simplicial complex, $\bd(\cC)^\vee \od \{ F^c \mid F \not\in\bd(\cC)\}$.
\end{lem}

We use this lemma to detect sphere link obstructions.

\begin{prop}\label{P:betti}
	Let $\cC$ be a \nd\ hyperplane code with polar complex $\bd(\cC)$. Then, \\ $\beta_{i,\s}(S/I_{\bd(\cC)^\vee}) = 0$ for all $i \geq 1$ except:
	\[ \begin{cases}
		\beta_{1,\S^c}(S/I_{\bd(\cC)^\vee}) = 1			& \text{if $\S$ is a facet}.\\
		\beta_{n-|\S|+1,\S^c}(S/I_{\bd(\cC)^\vee}) = 1	& \text{if $\link_\S\bd(\cC) \sim S^{n-|\S|-1}$}.
	\end{cases} \]
    \end{prop}

    \begin{proof}[Proof of \refP{betti}]

Inserting $i=0$ and $\S$ a facet 
into the dual version of Hochster's formula yields
\[\beta_{1,\S^c}(S/I_{\bd(\cC)^\vee}) = \dim_k \tilde{H}^{-1} (\link_{\S}{\bd(\cC)} ; k).\]
The right-hand side is equal to 1, since the link of a 
facet is the irrelevant simplicial complex, which gives a generator of 
$(-1)$-homology. This gives the first equation from the Proposition.

Setting $i = n - |\S|$ and $\S$
a face of $\bd(\cC)$:
\[\beta_{n-|\S|+1,\S}(S/I_{\bd(\cC)^\vee}) = \dim_k \tilde{H}^{n-|\S|-1} (\link_{\S^c}{\bd(\cC)} ; k).\]
The right-hand side is $1$ precisely when the link is a sphere of the
right dimension. 
In all other cases, the link is collapsible 
(Proposition \ref{Prop:37}) or equal to the void complex (links of non-faces), so the reduced homology is zero.
\end{proof}

\noindent This proposition provides an algebraic signature of \nd\ hyperplane codes.

\begin{ex}
	We again consider the code from \refEx{code3}. First, we translate into its
	polar complex $\bd(\cC_3)$, which has eleven facets for its eleven codewords. Then
	we compute the Stanley-Reisner ideal of its Alexander dual, and the Betti numbers
	associated to a minimal free resolution (e.g.\ using {\tt Macaulay2}).

	The table below is a condensed representation of the Betti numbers of $I_{\bd(\cC_3)}$, where the $(i,j)$-th entry is $\beta_{j,i+j}$ under the coarse grading.
	\[ \begin{array}{|r|R{5mm}R{5mm}R{5mm}R{5mm}R{5mm}|}%
	     \hline
	\mbox{\backslashbox[8mm]{\small$i$}{\small$j$}} & 0 & 1 & 2 & 3 & 4 \\ \hline
	0 & 1 &  &   &   &   \\
	1 &  &  &   &   &   \\
	2 &  &  &   &   &   \\
	3 &  & 11 & 16  & 6  &   \\
	4 &  &    & 1   & 2 &\\
	5 &  &    &     &   & 1 \\ \hline  \end{array} \]
	The value of $\beta_{1,4}$ counts the codewords, which are facets of $\bd(\cC)$.
	The remaining entries of row 3 indicate links with the appropriate dimension.
	Rows 4 and 5, under the multigrading, point to the following nonzero Betti numbers:
		\[
		\beta_{2,234\overline{134}} = 1, \hspace{1cm} \beta_{3,234\overline{1234}} = 1, 
		\hspace{1cm} \beta_{3,1234\overline{134}} = 1, \hspace{1cm} \beta_{4,1234\overline{1234}} = 1.
		\]
	Note that the multigrading of each Betti number corresponds to the link of its complement; specifically,
	$234\overline{134} \mapsto 1\bar{2}, 234\overline{1234} \mapsto 1, 1234\overline{134}\mapsto \bar{2}$, and $1234\overline{1234} \mapsto \varnothing$.
	These entries give us the following sphere link obstructions to $\bd(\cC_3)$ being the polar complex of a \nd\ hyperplane code.
	\begin{enumerate}
		\item $\link_{1\bar{2}}\bd(\cC_3) = \D(\{3\bar{4},\bar{3}4\})$, which has two connected components and hence nontrivial reduced homology of rank 1.
		\item $\link_{1}\bd(\cC_3) = \D(\{2\overline{34},\bar{2}3\bar{4}, \overline{23}4,23\bar{4},2\bar{3}4\}) \sim S^1$, which has the wrong dimension.
		\item $\link_{\bar{2}}\bd(\cC_3) = \D(\{\overline{134},\bar{1}3\bar{4}, \overline{13}4,13\bar{4},1\bar{3}4\}) \sim S^1$, which also has the wrong dimension.
		\item $\link_\varnothing \bd(\cC_3) = \bd(\cC_3)$ has nontrivial homology, but $\cC_3 \neq 2^{[4]}$.
	\end{enumerate}

	\noindent Each of these indicates the presence of a sphere link obstruction. 
	Thus, $\cC$ cannot be a {\nd} hyperplane code.

\end{ex}


\section{Proofs of \refPp{nondegenerateshelling} and \refT{shellableimplies}}\label{S:proofs}

\def\HS{H_{sw}}
\def\fs{f_{sw}}
\def\ms{m_{sw}}
\def\us{w_{sw}}
\def\cL{\mathcal L}
\def\cP{\mathcal P}
\subsection{Shellability}\label{S:proofs:shellability}
	The proof of \refPp{nondegenerateshelling} is organized as follows. First, we prove it in the special case $X = \R^d$. To extend the proof to the general case, we prove \nd\ hyerplane codes can be realized by a pair $(\cH,\cP)$ with $\cP$ the interior of a convex polyhedron with bounding hyperplanes $\cB$ such that $(\cH\cup\cB,\R^d)$ is a \nd\ arrangement. Lastly, we use links to consider $\cP$ as a region in $\R^d$, reducing to the special case.
		
		To prove the special case of \refPp{nondegenerateshelling}, we use the following equivalent definition of a shelling order (see, for example, \cite[Chapter III]{stanleyCCA}).

	\begin{defn}\label{D:shellability} 
		Let $\D$ be a simplicial complex and $F_1,\dots,F_t$ an ordering of its facets. The ordering is a shelling order if the sequence of complexes $\D_i = \D(\{F_1,\dots,F_i\})$, for each $i=2,\dots,t$, satisfies the property that the collection of faces $\D_i\sm\D_{i-1}$ has a \emph{unique} minimal element, denoted $r(F_i)$
			and called the \emph{associated minimal face of $F_i$}.
	\end{defn}
	
	\begin{lem}\label{L:Rdshell}
		If $(\cH,\R^d)$ has \subnd\ intersections, then $\bd(\code(\cH,\R^d))$ is shellable.
	\end{lem}
	\begin{proof}
		\def\HS{H}
		\def\HSm{H^-}
		\def\fs{f}
		\def\ms{m}
		\def\us{u}
		\def\cL{\mathcal L}
		\def\pts{\Omega} 
		Let $\cC = \code(\cH,\R^d)$ with $k = |\cC|$ the number of codewords. Without loss of generality, the $w_i$ defining the hyperplanes $H_i$ are unit vectors that span $\R^d$. Recall the notation
			\[ R_F = \bigcap_{i\in F} H_i^+ \cap \bigcap_{\bar{j}\in F} H_j^- \]
		for $F \in \bd(\cC)$.
		Our proof proceeds by induction on $d$, the ambient dimension. An example of the $d=2$ case is illustrated in \refF{shellingproof}.
		
	\begin{figure}[b]
		\begin{center}
			\begin{tabular}{l @{\hspace{2cm}} l}
				(a) \raisebox{-0.5\height}{\includegraphics[height=3cm,page=1]{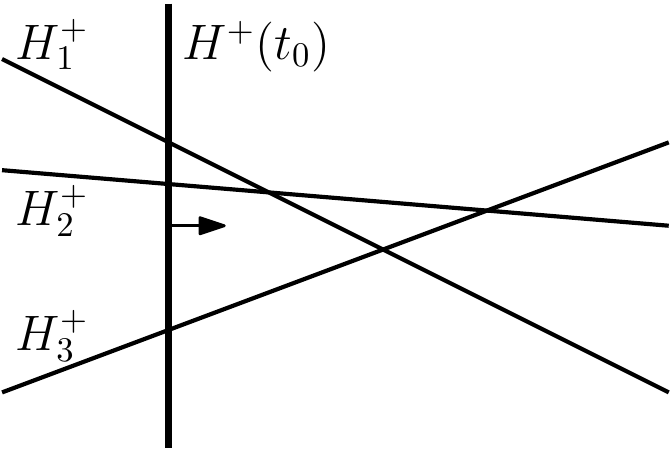}}
				& (d) \raisebox{-0.5\height}{\includegraphics[height=3cm,page=4]{polarcomplex_examples.pdf}}\\
				(b) \raisebox{-0.5\height}{\includegraphics[height=3cm,page=2]{shelling-proof.pdf}}
				& (e) \raisebox{-0.5\height}{\includegraphics[height=3cm,page=5]{polarcomplex_examples.pdf}}\\
				(c) \raisebox{-0.5\height}{\includegraphics[height=3cm,page=3]{shelling-proof.pdf}}
				& (f) \raisebox{-0.5\height}{\includegraphics[height=3cm,page=6]{polarcomplex_examples.pdf}}
			\end{tabular}
			\caption{An example of the shelling order construction in the $d = 2$ case.
			(a)~The atoms discovered at time $t_0$, i.e.\ the atoms $A_\s$ with $\ms(\s) = -\infty$. Note the four atoms of $(\cH,\R^2)$ which intersect $\HS(t_0)$ partition it into four intervals. 
			(b)~As $t$ increases, $\HS(t)$ slides to the right, encountering atoms one at a time. The shaded atom is newly discovered.
			(c)~Uniqueness of $r(\S(\s_6))$ follows because $\mathbf{e}_3\cdot\s_6$ and $\mathbf{e}_1\cdot\s_6$ have already been discovered. 
			(d),(e),(f)~The inductive step and next two steps of the shelling order. The associated minimal face is highlighted with a large mark (panel~(d)) or a dashed line (panels~(e),(f)).
			(d)~The polar complex $\bd(\code(\cL,\HS(t_0)))$. Ordering the four codewords discovered in panel~(a) from top to bottom yields $r(\br{1}2\br{3}) = \br{3}$. 
			(e)~Facet $123$ is added when $\HS(t)$ contains the intersection $H_1\cap H_2$ (panel (b)), thus $r(123) = 12$. 
			(f)~Atom $A_{12}$ is discovered when $\HS(t)$ contains $H_1\cap H_3$ (panel~(c)). Thus $r(12\br{3}) = 1\br{3}$.} 
			\label{F:shellingproof}
		\end{center}
	\end{figure}
		
		The base case $d = 1$ is straightforward and guides the intuition for the general case. We order the codewords of $\cC$ in a natural way based on their atoms, and show the corresponding ordering of facets of $\bd(\cC)$ is a shelling order. Each half-space $H_i^+$ is defined by an inequality of the form $x > h_i$ or $-x > h_i$ (i.e.\ $w_i = \pm1$ for all $i$). Each atom $A_\s$ has nonempty interior $(a_\s,b_\s)$ with $a_\s = h_{i_\s}$ for some $i_\s \in [n]$, with one exception, where $a_\s = -\infty$. Order the codewords $\s_1,\dots,\s_k$ in increasing order of $a_\s$. This is a shelling order: when we add facet $\S(\s)$ to our simplicial complex, this is the first time a facet contains $i_{\s}$ if $w_i = 1$, otherwise it's the first time a facet contains $\br{i}_\s$. In other words, $\s$ is the first codeword in this order which contains $i$ if $w_i = 1$ or the first codeword which does not contain $i$ if $w_i = -1$; all later atoms lie on the same side of the hyperplane $H_i$. See \refF{shellingproof}(d). Thus, every facet of $\bd(\cC)$ has an associated minimal face and this ordering is a shelling order.
		
		Now consider $d > 1$. Denote by $\pts(\cH)$ the set of points where $d$ hyperplanes intersect. We choose a generic ``sweep'' direction, a vector $\us \in \R^d$ which satisfies the following  properties:
			\begin{enumerate}[label={(\roman*)},leftmargin=*]
				\item\label{property:independent} $\us$ is not in the span of any $(d-1)$-element subset of $\{w_1,\dots,w_n\}$.
				\item\label{property:injective} For every pair of distinct points $x,y$ in $\pts(\cH)$, $\us$ is not in the orthogonal complement $(x - y)^\perp$.
			\end{enumerate}
		Such a $\us$ exists because we exclude finitely many subsets of measure zero from $\R^d$. We use $\us$ to define a sliding hyperplane $\HS(t)$ and its corresponding ``discovery time'' function $\ms:\cC \to \R\cup\{-\infty\},$
		\begin{align*}
			\HS(t) & = \{x\in\R^d \mid \us\cdot x - t = 0\}\\
			\ms(\s) & = \inf\{\us\cdot x \mid x \in A_\s\}.
		\end{align*}
		In the $d=1$ case, $m(\s)= a_\s$ and thus induces a total order on codewords.
		For the $d>1$ case, the goal is once again to use $\ms$ to order the codewords.
		To do this,
		\textbf{(1)}~we order the codewords with $\ms(\s) = -\infty$ inductively, then 
		\textbf{(2)}~we show $\ms$ is injective on the remaining codewords, and lastly, 
		\textbf{(3)}~we show every facet has an associated minimal face.

		\textbf{(1)} By construction, $\cH \cup \{\HS^+(t)\}$ is a \nd\ arrangement in $\R^d$ for all but finitely many values of $t$, specifically, the values where $\HS(t)$ contains a point in $\pts(\cH)$. Let $t_0$ be a constant less than all of these values (see \refF{shellingproof}(a) for an illustration). Property \ref{property:independent} ensures $H_i^+ \cap \HS(t_0) \neq \varnothing$ for all $i$, so in particular $\cL \od \{H_i^+ \cap \HS(t_0)\}$ is a {\nd} arrangement in $\HS(t_0) \cong \R^{d-1}$. By inductive hypothesis, $\bd(\code(\cL, \HS(t_0)))$ is shellable. Each nonempty atom of the arrangement $(\cL, \HS(t_0))$ is the intersection of an atom of $(\cH,\R^d)$ with $\HS(t_0)$, and the corresponding codewords are precisely those with $\ms(\s) = -\infty$. Thus, we have an ordering for these codewords which is an initial segment of a shelling of $\bd(\cC)$ (\refF{shellingproof}(d)).
		
		\textbf{(2)} Let $\s\in\cC$ be a codeword with $\ms(\s) > -\infty$. The function $f(x) = \us\cdot x$ is minimized along a face of the (closure of) polyhedron $R_{\S(\s)}$; property \ref{property:independent} ensures this face is a vertex, which is an element of $\pts(\cH)$. Property \ref{property:injective} ensures $f|_{\pts(\cH)}$ is injective. Therefore, $\ms$ induces a total order on codewords $\s$ with $\ms(\s) > -\infty$. Let $\s_1,\dots,\s_k$ be the ordering of codewords of $\cC$ obtained appending this ordering to the order from \textbf{(1)}. We will show each facet has an associated minimal face to complete the proof.
		
		\textbf{(3)} Denote $\bd_i \od \bd(\{\s_1,\dots,\s_i\})$ for $i=1,\dots,k$. From \textbf{(1)}, $r(\S(\s_i))$ is defined whenever $\ms(\s_i) = -\infty$. So, let $\s_i$ be a codeword with with $t_i = \ms(\s_i) > -\infty$, meaning there is a vertex of $R_{\S(\s_i)}$ minimizing $f$. This vertex is an element of $\pts(\cH)$, i.e.\ it is the intersection $H_{\a_i}$ of $d$ hyperplanes (see \refF{shellingproof}(b) and (c)).
		For $F\in\bd(\cC)$ and $\a\subseteq[n]$, we denote
			\[ F|\a \od F \cap (\a \sqcup \br{\a}), \]
		the subset of $F$ with support $\a$.
		We claim $r(\S(\s_i)) = \S(\s_i)|\a_i$ 
		(see \refF{shellingproof}(e) and (f)).
		The region $R_{\S(\s_i)|\a_i}$ is a cone supported by $\HS(t_i)$, so this is the first codeword in our order with this exact combination of ``on'' and ``off'' vertices indexed by $\a_i$. Thus, $\S(\s_i)|\a_i \in \bd_i \sm \bd_{i-1}$.
		
		Now consider $F = \S(\s_i)|\b \in \bd_i \sm \bd_{i-1}$. Suppose, for the sake of contradiction, $\b \not\supseteq \a_i$, that is, there is some $\ell \in \a_i\sm\b$. Then $F \subseteq \S(\mathbf{e}_\ell \cdot \s_i)$. Note $\mathbf{e}_\ell \cdot\s_i \in \cC$ since, by \subndy, all $2^d$ possible regions around the point $H_{\a_i}$ produce codewords. However, since $\HS(t_i)$ intersects the interior of $R_{\S(\mathbf{e}_\ell \cdot\s_i)}$, we have $\ms(\mathbf{e}_\ell \cdot\s_i) < \ms(\s_i)$ and therefore $\S(\mathbf{e}_\ell \cdot \s_i) \in \bd_{i-1}$. We reach a contradiction, as this implies $F \in \bd_{i-1}.$
		Therefore, $r(\S_i) = \S_i|\a_i$ is the unique minimal face in $\bd_i \sm \bd_{i-1}$. This completes the proof.
	\end{proof}

	We now prove that a \nd\ hyperplane code is a subset of codewords of a \nd\ hyperplane arrangement in $\R^d$.
		
	\begin{lem}\label{L:polyhedralapproximation} 
		If $\cC$ is a \nd\ hyperplane code, then $\cC$ can be realized by a \nd\ pair $(\cH, \mathcal P )$ such that $ \mathcal P = \bigcap_{j\in[m]}B_j^+$ is an open polytope with bounding hyperplanes $\cB$ such that $\cH\cup\cB$ has \subnd\ intersections in $\R^d$.
	\end{lem}
	\begin{proof}
		Let $(\cH,X)$ be a \nd\ pair realizing $\cC$.
		By \refL{nondegenerateatoms}, we can perturb the hyperplanes $\cH$ to an arrangement $\cH'$ while preserving the atoms of the arrangement $(\cH,X)$, i.e.\ $\code(\cH',X) = \code(\cH,X)$.
		Thus, $\cC$ has a realization $(\cH',X)$ such that $\cH'$ has \subnd\ intersections outside of $X$ as well.
		
		Applying \refL{nondegenerateatoms} again, we can choose a point $p_\s$ in the interior of $A_\s^{\cH'}$ for every $\s\in\cC$.
		Let $\cP$ be the interior of the convex hull of the set of points $\{p_\s \mid \s\in\cC\}$; by perturbing the points slightly we may assume $\cP$ is full-dimensional.
		Let $\cB =\{B_{n+1}^+,\dots,B_{n+m}^+\}$ denote the bounding hyperplanes of this polytope, i.e.\  $\cP = \bigcap_{j=n+1}^{n+m} B_j^+$.
		Since $\cP\subseteq X$, we conclude  $\code(\cH',\cP) \subseteq \code(\cH',X)$. 
		Since we chose a points $p_\s$ for every codeword of $\C$,  $\s\in\cC $ implies $ A_\s^{\cH'}\cap \cP \neq \varnothing$  and therefore $\code(\cH',X) \subseteq \code(\cH',\cP)$.
		Thus we have $\cC = \code(\cH',\cP)$ and $(\cH',\cP)$ is a \nd\ arrangement.
		
		The hyperplanes in $\cH'\cup\cB$ do not necessarily have \subnd\ intersections. Again, we apply \refL{nondegenerateatoms}: one can perturb each hyperplane in  $\cB$ to hyperplanes $\cB^\prime$, so that these hyperplanes have  \subnd\ intersections, yet the appropriate code  is preserved, i.e.   $\C= \code  (\cH' ,\cP )= (\cH', \cP^\prime)$, where  $ \cP^\prime$ is the open polyhedron $\cP' = \bigcap_{B\in\cB'} B^+$. This completes the proof.
	\end{proof}

	\noindent We extend \refL{Rdshell} to the general case with the following standard lemma \cite{BW1997}.
	
	\begin{lem}[{\cite[Proposition 10.14]{BW1997}}]\label{L:shellablelink}
		Let $\D$ be a shellable simplicial complex. Then $\link_\s\D$ is shellable for any $\s\in\D$, with shelling order induced from the shelling order of $\D$.
	\end{lem}

	\begin{proof}[Proof of \refPp{nondegenerateshelling}]
		By \refL{polyhedralapproximation}, $\cC$ can be realized as $\cC = \code(\cH,\mathcal P)$ with
			\[ \mathcal P = \bigcap_{j=n+1}^{n+m} B_j^+ \]
		an open polyhedron such that the arrangement $\cH \cup \cB$ has \subnd\ intersections in $\R^d$. Set $\cC' = \code(\cH\cup\cB,\R^d)$, a code on vertex set $[n+m]$. By \refL{Rdshell}, $\bd(\cC')$ is shellable. Set $F = \{n+1,\dots,n+m\} \in \bd(\cC')$. Then we have
			\[ \link_F \bd(\cC') = \bd\left(\code\Bigl(\cH,\bigcap_{j=n+1}^{n+m} B_j^+\Bigr)\right) = \bd(\cC). \]
		By \refL{shellablelink}, as the link of a shellable complex, $\bd(\cC)$ is shellable.
	\end{proof}

\subsection{Obstructions following from shellability}\label{S:proofs:shellabilityobstructions}

	In general, shellable simplicial complexes are homotopy-equivalent to a wedge sum of spheres, where the number and dimension of the spheres correspond to the facets with $r(F) = F$ in some shelling order \cite{Kozlov2008}. First we prove a stronger version of this statement for the polar complex of a code, which will be used throughout the proofs of all parts of \refT{shellableimplies}. 
	Note the condition of this lemma is intrinsic to the polar complex of the code and does not rely on any particular realization. 
	
	\begin{lem}\label{L:shellablecollapsible}
		If $\bd(\cC)$ is shellable, then either $\,\cC = 2^{[n]}$ or $\bd(\cC)$ is collapsible.
	\end{lem}
	
	\begin{proof}
		We induct on the number of codewords of $\cC$.
		Let $F_1,\dots,F_t$ be a shelling order of $\bd(\cC)$, with $\s_1,\dots,\s_t$ the corresponding order of codewords in $\cC$.
		For ease of notation, let $\cC' = \{\s_1,\dots,\s_{t-1}\}$ denote the first $t-1$ codewords in this shelling order.
		By construction, $\bd(\cC')$ is shellable.
		Because it has one fewer codeword than $\cC$, it cannot be the full code and therefore, by inductive hypothesis, $\bd(\cC')$ is collapsible.
		
		By definition, $r(F_t)$ is the unique minimal element of the collection $\bd(\cC) \sm \bd(\cC')$ and hence the only facet that contains $r(F_t)$ is $F_t$.
		If $r(F_t) \subsetneq F_t$, then $(r(F_t),F_t)$ is a free pair, and $\bd(\cC) \searrow_{\,r(F_t)} \bd(\cC')$ which is collapsible.
		
		In the case $r(F_t) = F_t$, we claim we must have $\cC = 2^{[n]}$.
		Suppose not, for the sake of contradiction, and let $\t \in 2^{[n]}\sm\cC$. 
		Note $\bd(2^{[n]}\sm\{\t\})$ is homeomorphic to a closed $(n-1)$-ball (as it is a sphere missing top-dimensional open disc).
		Since $\bd(\cC')$ is a collapsible subcomplex of a simplicial complex, $\bd(\cC)$ is homotopy-equivalent to the quotient space $\bd(\cC) / \bd(\cC')$ (see \cite[Proposition 0.17 and Proposition A.5]{Hatcher}).
		Because $r(F_t) = F_t$, the boundary of the simplex $\D(\{F_t\})$ is contained in $\bd(\cC')$, and therefore $\bd(\cC)/\bd(\cC')$ is homotopy equivalent to $S^{n-1}$.
		We reach a contradiction, as $\bd(\cC) \subseteq \bd(2^{[n]}\sm\{\t\})$, but there is no embedding $S^{n-1} \hookrightarrow \R^{n-1}$ (see, e.g.\ \cite[Corollary 2B.4]{Hatcher}).
		Therefore, in this case we have $\cC = 2^{[n]}$.
	\end{proof}
	
	\noindent To prove \refT{shellableimplies}.1, we need one more lemma. Note that this lemma concerns  with \emph{contractibility} of certain subcomplexes, hence it can only be used to show $\cC$ has no \emph{weak} local obstructions.
	
	\begin{lem}[{\cite[Lemma 4.4]{Carina 2015}}]\label{L:linklemma}
		Let $\D$ be a simplicial complex on vertex set $V$. Let $\a,\b\in\D$ with $\a\cap\b =\varnothing$, $\a\cup\b \subsetneq V$, and $\link_\a(\D|_{\a\cup\b})$ \emph{not} contractible. Then there exists $\a'\in\D$ such that (i) $\a'\supseteq \a$, (ii) $\a'\cap\b = \varnothing$, and (iii) $\link_{\a'}(\D)$ is not contractible.
	\end{lem}
	
	\begin{proof}[Proof of \refT{shellableimplies}.\ref{T:shellableimpliesnobflo}] Assume that the polar complex $\Gamma(\C)$ is shellable. 
		To show that $\cC$ has no weak local obstructions, first suppose $\t\in\D(\cC)$ and $\link_\t\D(\cC)$ is not contractible. We will show $\t\in\cC$. Note that $\D(\cC) = \bd(\cC)|_{[n]\sqcup\varnothing}$, thus  we apply \refL{linklemma} to the pair $\a = \t\sqcup\varnothing,$ $\b = ([n]\sm\t)\sqcup\varnothing$ in the polar complex $\bd(\cC)$: there exists a face $T \in \bd(\cC)$ such that (i) $T = T^+ \sqcup \br{T^-} \supseteq \t\sqcup\varnothing$, (ii) $T\cap(([n]\sm\t)\sqcup\varnothing) = \varnothing$, and (iii) $\link_T\bd(\cC)$ is not contractible.
	Statements (i) and (ii) together imply $T^+ = \t$. Statement (iii) together with \refL{shellablecollapsible}, implies $\link_T\bd(\cC) = \bd(2^{[n]\sm\underline{T}})$. Therefore this link contains the facet $F$   consisting of all barred vertices in $ [n]\sm\underline{T}$. Thus   $T \cup F = \t \sqcup \br{[n]\sm\t}$ is a face of $\bd(\cC)$ and therefore $\t\in\cC$; hence $\t$ cannot be a local obstruction.
		
		For any $g\in(\ZZ_2)^n$, the above argument extends to $g\cdot\cC$ verbatim, since $\Gamma(g\cdot \C)=g\cdot \Gamma( \C)$, and $g\cdot \Gamma( \C)$ is also shellable. Thus, $\cC$ has no bitflip local obstructions.
	\end{proof}
	
	\begin{proof}[Proof of \refT{shellableimplies}.\ref{T:shellableimpliesnoslo}]
		Links of $\bd(\cC)$ are polar complexes of a code on a smaller set of vertices, and links of shellable complexes are shellable (\refL{shellablelink}). Therefore, we can apply \refL{shellablecollapsible} to conclude $\link_F\bd(\cC)$ is either collapsible or $\bd(2^{[n]\sm\underline{F}})$ for any $F \in \bd(\cC)$. 		Thus, no face $F$ can be a sphere link obstruction.
	\end{proof}
	
	We use one final lemma to prove \refT{shellableimplies}.3, which concerns faces of simplicial complexes with collapsible links.
	
	\begin{lem}\label{L:link-collapse}
		Let $\D$ be a simplicial complex with $\a\in\D$ such that $\link_\a\D$ is collapsible. Then $\D \searrow \del_\a\D$.
	\end{lem}
	
	\begin{proof}
		Let $(\s_1,\t_1),\dots,(\s_k,\t_k)$ be the sequence of free pairs along which $\D_1 = \link_\a\D$ is collapsed (in particular, $\s_k = \varnothing$), resulting in the sequence of simplicial complexes
			\[ \link_\a\D = \D_1 \searrow_{\s_1} \D_2 \searrow_{\s_2} \cdots \searrow_{\s_k} \D_{k+1} = \{\}.\]
		Consider the sequence $(\s_1\cup\a,\t_1\cup\a),\dots,(\s_k\cup\a,\t_k\cup\a)$ in $\D$. We claim, $(\s_1\cup\a,\t_1\cup\a)$ is a free pair: $\s_1\cup\a \subsetneq \t_1\cup\a$ and $\t_1\cup\a$ is a facet of $\D$. If $\s_1\cup\a\subseteq\t'$ for some facet $\t'$, then $\t'\sm\a$ is a facet of $\link_\a\D$ which contains $\s_1,$ hence $\t' = \t$. This argument can be repeated for the pair $(\s_2\cup\a,\t_2\cup\a)$ in $\del_{\s_1\cup\a}\D$, and so on, to show that this is a sequence of free pairs in $\D$. Thus, we have a sequence of collapses
			\[ \D \searrow_{\s_1\cup\a} \cdots \searrow_{\s_k\cup\a} \del_{\s_k\cup\a}\D. \]
		Since $\s_k\cup\a = \a$, we have $\D \searrow \del_\a\D$.
	\end{proof}
		
	\begin{proof}[Proof of \refT{shellableimplies}.\ref{T:shellableimpliesnoco}]
		Assume the polar complex $\bd(\cC)$ is shellable. 
		We demonstrate that if $\s\in\cham(\cC)$ has more than one chamber, then $\s$ is not maximal.
		
		Suppose $T_1 \neq T_2$ are chambers of $\s$, that is
			\[ \link_{T_1}\bd(\cC) = \link_{T_2}\bd(\cC) = \bd(2^\s). \]
		We will proceed by induction on $k = |T_1\sm T_2|>0$. Since $\underline{T_1} = \underline{T_2} = [n]\sm\s$, $k$ is the number of indices where one $T_i$ has a barred vertex and the other does not.
		
		For the base case $k=1$, suppose $T_1 \sm T_2 = i$.
		Then \[ \link_{T_1\cap T_2}\bd(\cC) = \bd(2^{\s\cup \{i\}}) \] so $\s\cup i \in \cham(\cC)$ and $\s$ is not maximal.
			
			Now suppose $|T_1\sm T_2| = k > 1$. 
		We produce a face $F$ such that
		$\link_F\bd(\cC) = \bd(2^\s)$ and $|T_1\sm F| < k$, 
		giving the induction step.  
		Let $T = T_1\cap T_2,$ and consider $\link_T\bd(\cC).$ 
		This is a shellable subcomplex of $\bd(2^{[n]\sm\underline{T}})$; 
		denote its corresponding code by $\cC'$. 
		Let $T_1' = T_1\sm T$ and $T_2' = T_2\sm T$; by design these are disjoint with $|T_1' \sm T_2'| = |T_1'| = |T_2'| = k$ and 
		$\link_{T'_i}\bd(\cC') = \bd(2^\s)$ for $i=1,2$. Because they are disjoint, $\operatorname{star}_{T_1'}\bd(\cC') \cup \operatorname{star}_{T_2'}\bd(\cC')$ is a suspension of $\bd(2^\s)$, making it homotopy equivalent to $S^{|\s|}$.
					
		Consider a face $F' \in \bd(\cC')$ such that $\underline{F'} = \underline{T_1'}$. By construction, $\link_{F'}\bd(\cC')$ is a subcomplex of $\bd(2^\s)$. If $\link_{F'}\bd(\cC') \neq \bd(2^\s)$, then the link is collapsible by Lemmas~\ref{L:shellablelink} and \ref{L:shellablecollapsible}; \refL{link-collapse} implies that $\bd(\cC')$ collapses to $\del_{F'}\bd(\cC')$.
		
		There are $2^k - 2$ faces $F'\neq T_1,T_2$ with 
		$\underline{F'} = \underline{T_1}$. 
		If none of these $F'$ had $\link_{F'}\bd(\cC') = \bd(2^\s)$,
		this would lead to a contradiction: 
		we would have a sequence of collapses
		\[ \bd(\cC') \searrow \operatorname{star}_{T_1'}\bd(\cC') \cup \operatorname{star}_{T_2'}\bd(\cC'). \]
		Since $\bd(\cC)$ is shellable, 
		by \refL{shellablecollapsible} it
		is homotopy equivalent to $S^{n-1}$ or is contractible.
		Collapsing preserves homotopy type, so we reach a contradiction.
				
		Therefore, for one of these $F'$ we must have 
		$\link_{F'}\bd(\cC') = \bd(2^\s)$. 
		Thus $\link_{F'\cup T}\bd(\cC) = \bd(2^\s)$ and so we have another face in $\bd(\cC)$ whose link yields $\bd(2^\s)$, namely $F = F'\cup T$. 
		Since $|T_1 \sm F| < k$, by induction $\s$ is not maximal in $\cham\cC$. Therefore, if $\s$ is maximal in $\cC$, it must have a unique chamber, and thus $\cC$ has no chamber obstructions.
	\end{proof}

\addtocontents{toc}{\protect\setcounter{tocdepth}{0}}

\section*{Acknowledgments}
This work was  supported by the joint NSF DMS/NIGMS grant R01GM117592,
 NSF IOS-155925  to VI.   
Research by ZR was partially supported by a Math+X 
Research Grant.
The authors would like to thank Carina Curto, 
Art Duval, Jack Jeffries, Katie Morrison and 
Anne Shiu for helpful discussions.

\bibliography{references.bib}

\begin{bibdiv}
\begin{biblist}

\bib{COM}{article}{
      author={Bandelt, Hans-J{\"u}rgen},
      author={Chepoi, Victor},
      author={Knauer, Kolja},
       title={{COMs}: {Complexes} of oriented matroids},
        date={2018-05},
        ISSN={0097-3165},
     journal={Journal of Combinatorial Theory, Series A},
      volume={156},
       pages={195\ndash 237},
  url={http://www.sciencedirect.com/science/article/pii/S0097316518300049},
}

\bib{AOM}{article}{
      author={Baum, Andrea},
      author={Zhu, Yida},
       title={The axiomatization of affine oriented matroids reassessed},
    language={en},
        date={2018-01},
        ISSN={1420-8997},
     journal={Journal of Geometry},
      volume={109},
      number={1},
       pages={11},
         url={https://doi.org/10.1007/s00022-018-0407-5},
}

\bib{Bjorner:1996}{incollection}{
      author={Bj\"orner, Anders},
       title={Handbook of combinatorics},
        date={1995},
      editor={Graham, R.~L.},
      editor={Grotschel, M.},
      editor={Lovasz, L.},
   publisher={MIT Press},
     address={Cambridge, MA, USA},
       pages={1819\ndash 1872},
}

\bib{OM}{book}{
      author={Bj{\"o}rner, Anders},
      author={Las~Vergnas, Michel},
      author={Sturmfels, Bernd},
      author={White, Neil},
      author={Ziegler, G{\"u}nter~M},
       title={Oriented matroids},
      series={Encyclopedia of Mathematics and its Applications},
   publisher={Cambridge University Press},
        date={1999},
      volume={46},
}

\bib{BW1997}{article}{
      author={Bj\"orner, Anders},
      author={Wachs, Michelle~L.},
       title={Shellable nonpure complexes and posets. {II}},
        date={1997},
     journal={Trans. Amer. Math. Soc.},
      volume={349},
      number={10},
       pages={3945\ndash 3975},
}

\bib{shiu2018}{article}{
      author={Chen, Aaron},
      author={Frick, Florian},
      author={Shiu, Anne},
       title={Neural codes, decidability, and a new local obstruction to
  convexity},
        date={2018-03},
     journal={arXiv:1803.11516v1 [math.CO]},
        note={Available at \url{http://arxiv.org/abs/1803.11516v1}},
}

\bib{cruzetal}{article}{
      author={Cruz, Joshua},
      author={Giusti, Chad},
      author={Itskov, Vladimir},
      author={Kronholm, Bill},
       title={On open and closed convex codes},
        date={2018},
        ISSN={1432-0444},
     journal={Discrete {\&} Computational Geometry},
         url={https://doi.org/10.1007/s00454-018-00050-1},
}

\bib{Carina2015}{article}{
      author={Curto, Carina},
      author={Gross, Elizabeth},
      author={Jeffries, Jack},
      author={Morrison, Katherine},
      author={Omar, Mohamed},
      author={Rosen, Zvi},
      author={Shiu, Anne},
      author={Youngs, Nora},
       title={What makes a neural code convex?},
        date={2017},
     journal={SIAM Journal on Applied Algebra and Geometry},
      volume={1},
      number={1},
       pages={222\ndash 238},
}

\bib{neuralring}{article}{
      author={Curto, Carina},
      author={Itskov, Vladimir},
      author={Veliz-Cuba, Alan},
      author={Youngs, Nora},
       title={The neural ring: an algebraic tool for analyzing the intrinsic
  structure of neural codes.},
        date={2013},
     journal={Bulletin of mathematical biology},
      volume={75},
      number={9},
       pages={1571\ndash 1611},
}

\bib{homomorphisms}{article}{
      author={Curto, Carina},
      author={Youngs, Nora},
       title={Neural ring homomorphisms and maps between neural codes},
        date={2015-11},
     journal={arXiv:1511.00255v2 [q-bio.NC]},
        note={Available at \url{http://arxiv.org/abs/1511.00255v2}},
}

\bib{Cybenko89}{article}{
      author={Cybenko, G},
       title={Approximations by superpositions of sigmoidal functions},
        date={1989},
     journal={Mathematics of Control, Signals, and Systems},
      volume={2},
      number={4},
       pages={303\ndash 314},
}

\bib{garcia}{article}{
      author={Garcia, Rebecca},
      author={Puente, Luis David~Garc{\'\i}a},
      author={Kruse, Ryan},
      author={Liu, Jessica},
      author={Miyata, Dane},
      author={Petersen, Ethan},
      author={Phillipson, Kaitlyn},
      author={Shiu, Anne},
       title={Gr{\"o}bner bases of neural ideals},
        date={2018-03},
        ISSN={0218-1967},
     journal={International Journal of Algebra and Computation},
      volume={28},
      number={04},
       pages={553\ndash 571},
  url={https://www.worldscientific.com/doi/abs/10.1142/S0218196718500261},
}

\bib{GI14}{article}{
      author={Giusti, Chad},
      author={Itskov, Vladimir},
       title={A no-go theorem for one-layer feedforward networks},
        date={2014},
     journal={Neural Comput.},
      volume={26},
      number={11},
}

\bib{Shellability_is_NP_Complete}{inproceedings}{
      author={Goaoc, Xavier},
      author={Pat{\'a}k, Pavel},
      author={Pat{\'a}kov{\'a}, Zuzana},
      author={Tancer, Martin},
      author={Wagner, Uli},
       title={{Shellability is NP-Complete}},
        date={2018},
   booktitle={34th international symposium on computational geometry (socg
  2018)},
      editor={Speckmann, Bettina},
      editor={T{\'o}th, Csaba~D.},
      series={Leibniz International Proceedings in Informatics (LIPIcs)},
      volume={99},
       pages={41:1\ndash 41:15},
}

\bib{M2}{misc}{
      author={Grayson, Daniel~R.},
      author={Stillman, Michael~E.},
       title={Macaulay2, a software system for research in algebraic geometry},
         how={Available at \url{http://www.math.uiuc.edu/Macaulay2/}},
        note={Available at \url{http://www.math.uiuc.edu/Macaulay2/}},
}

\bib{Jeffries}{article}{
      author={G\"unt\"urk\"un, Sema},
      author={Jeffries, Jack},
      author={Sun, Jeffrey},
       title={Polarization of neural rings},
        date={2017-06},
     journal={arXiv:1706.08559 [math.AC]},
        note={Available at \url{http://arxiv.org/abs/1706.08559}},
}

\bib{entorhinal}{article}{
      author={Hafting, Torkel},
      author={Fyhn, Marianne},
      author={Molden, Sturla},
      author={Moser, May-Britt},
      author={Moser, Edvard~I.},
       title={Microstructure of a spatial map in the entorhinal cortex.},
        date={2005},
     journal={Nature},
      volume={436},
      number={7052},
       pages={801\ndash 806},
}

\bib{Hatcher}{book}{
      author={Hatcher, Allen},
       title={Algebraic topology},
   publisher={Cambridge University Press},
     address={Cambridge},
        date={2002},
}

\bib{Hornik91}{article}{
      author={Hornik, Kurt},
       title={Approximation capabilities of multilayer feedforward networks},
        date={1991},
     journal={Neural Networks},
      volume={4},
      number={2},
       pages={251\ndash 257},
}

\bib{visual}{article}{
      author={Hubel, David~H},
      author={Wiesel, Torsten~N},
       title={Receptive fields, binocular interaction and functional
  architecture in the cat's visual cortex},
        date={1962},
     journal={The Journal of physiology},
      volume={160},
      number={1},
       pages={106\ndash 154},
}

\bib{jeffsCUR}{article}{
      author={Jeffs, R.~Amzi},
      author={Novik, Isabella},
       title={Convex {Union} {Representability} and {Convex} {Codes}},
        date={2018-08},
     journal={arXiv:1808.03992 [math.CO]},
        note={Available at \url{http://arxiv.org/abs/1808.03992}},
}

\bib{Kozlov2008}{book}{
      author={Kozlov, Dmitry},
       title={Combinatorial algebraic topology},
      series={Algorithms and Computation in Mathematics},
   publisher={Springer, Berlin},
        date={2008},
      volume={21},
}

\bib{CCA}{book}{
      author={Miller, Ezra},
      author={Sturmfels, Bernd},
       title={Combinatorial commutative algebra},
   publisher={Springer Science \& Business Media},
        date={2004},
      volume={227},
}

\bib{hippocampus}{article}{
      author={O'Keefe, John},
       title={Place units in the hippocampus of the freely moving rat},
        date={1976},
     journal={Experimental neurology},
      volume={51},
      number={1},
       pages={78\ndash 109},
}

\bib{R62}{book}{
      author={Rosenblatt, Frank},
       title={Principles of neurodynamics; perceptrons and the theory of brain
  mechanisms},
   publisher={Spartan Book},
        date={1962},
}

\bib{stanleyCCA}{book}{
      author={Stanley, Richard~P.},
       title={Combinatorics and commutative algebra},
     edition={Second Edition},
   publisher={Birkh\"auser},
        date={1996},
}

\end{biblist}
\end{bibdiv}

\end{document}